\documentclass{article}

\usepackage{times}

\newcommand{\Con}[1]{\mathbf{Cn}\left(#1\right)}

\newcommand{\acti}{\mathbf{A}}
\newcommand{\inhib}{\mathbf{I}}
\newcommand{\size}[1]{||#1||}

\newcommand{\migtuple}
    {\langle At\comma Ex\comma Ed\comma  \RPr\comma \RPrc\comma
      \RAct\comma \RInh\comma {\cal B}\rangle} 

\usepackage{graphicx}
\usepackage[caption=false]{subfig}
\usepackage{framed}
\usepackage{soul} 
\usepackage{url}

\newcommand{\red}[1]{{\color{red}{#1}}}

\newenvironment{comment}
  {\begin{framed}\noindent}
  {\end{framed}}

\usepackage{amsmath,amssymb,amsthm}
\usepackage[inline]{enumitem}
\usepackage[normalem]{ulem}
\usepackage{color}

\DeclareSymbolFont{bskarr}{U}{bskarr}{m}{n}
\DeclareFontFamily{U}{bskarr}{ }
\DeclareFontShape{U}{bskarr}{m}{n}{<->bskarr10}{}
\DeclareMathSymbol{\rightarrowTriangle}{\mathrel}{bskarr}{"FB}
\DeclareMathSymbol{\rightarrowtriangle}{\mathrel}{bskarr}{"FE}

\DeclareSymbolFont{arrows}{U}{FdSymbolC}{m}{n}
\DeclareFontFamily{U}{FdSymbolC}{}

\DeclareFontShape{U}{FdSymbolC}{m}{n}{
    <-7.1> s * [1.0] FdSymbolC-Book
    <7.1-> s * [1.0] FdSymbolC-Book
}{}
\DeclareFontShape{U}{FdSymbolC}{b}{n}{
    <-7.1> s * [1.0] FdSymbolC-Medium
    <7.1-> s * [1.0] FdSymbolC-Medium
}{}

\DeclareMathSymbol{\longrightfootline}{\mathrel}{arrows}{173}

\newcommand{\inftr}[1]{ {\langle {#1} \rangle}}

\def\Next{\bigcirc}
\newcommand{\Nat}{\mathbb{N}}
\def\eqdef{\stackrel{\rm def}{=}}

\def\Aprod{\rightarrowTriangle}
\def\Ainhib{\longrightfootline}

\def\Aactiv{\rightarrowtriangle}

\newcommand{\Prod}[1]{\mathbf{Pr}\left(#1\right)}
\newcommand{\hide}[1]{}
\newcommand{\imp}{\rightarrow}

\newcommand{\G}{\mbox{$\cal G$}}
\def\RInh{\stackrel{\longrightfootline}{{\cal R}}}
\def\RPr{\stackrel{\rightarrowtriangle}{{\cal P}}}
\def\RPrc{\stackrel{\rightarrowTriangle}{{\cal P}}}

\def\RAct{{{\cal R}_{\rightarrowtriangle}}}
\def\RInh{{{\cal R}_{\longrightfootline}}}
\def\RPr{{{\cal P}_{\rightarrowtriangle}}}
\def\RPrc{{{\cal P}_{\rightarrowTriangle}}}

\def\RAct{{\cal A}}
\def\RInh{{\cal I}}
\def\RPr{{\cal P}}
\def\RPrc{{\cal C}}

\newcommand{\comma}{,\allowbreak}

\newtheorem{definition}{Definition}
\newtheorem{lemma}{Lemma}
\newtheorem{theorem}{Theorem}

\newtheorem{example}{Example}
\newtheorem{remark}{Remark}

\renewcommand{\red}[1]{#1}

\newcommand{\tocheck}{}

\title{A framework for modelling Molecular Interaction Maps}

\date{}

\author{Jean-Marc Alliot%
\footnote{Institut de Recherche en Informatique de Toulouse. Toulouse
University, Toulouse, France}\\
\texttt{\small jean-marc.alliot@irit.fr}
\and
Marta Cialdea Mayer%
\footnote{Dipartimento di Ingegneria, Universit\`a degli Studi Roma Tre, Rome,
Italy}\\
\texttt{\small cialdea@ing.uniroma3.it}
\and
Robert Demolombe$^*$\\
\texttt{\small demolombe@irit.fr}\\
\and
Mart\'in Di\'eguez$^*$\\
\texttt{\small martin.dieguez@irit.fr}\\
\and
Luis Fari\~nas del Cerro$^*$\\
\texttt{\small farinas@irit.fr}
}

\begin{document}
\maketitle
\begin{abstract}
Metabolic networks, formed by a series of metabolic pathways, are made
of intracellular and extracellular reactions that determine the
biochemical properties of a cell, and by a set of interactions that
guide and regulate the activity of these reactions. 
 Most of these pathways are formed by an
intricate and complex network of chain reactions, and can be
represented in a human readable form using graphs which describe the
cell cycle checkpoint pathways.  

This paper proposes a method to represent Molecular Interaction Maps (graphical
representations of complex metabolic networks) in Linear Temporal Logic. The logical
representation of such networks allows one to reason about them, in order to check,
for instance, whether a graph satisfies a given property  $\phi$, as well as to find out 
which initial conditons would guarantee $\phi$, or else how can the 
the graph be updated 
  in order to satisfy  $\phi$.  

 Both the translation and resolution methods have been implemented in a
tool capable of addressing such questions thanks to a reduction to propositional logic
which allows exploiting classical SAT solvers.

\end{abstract}

\newpage
\section{Introduction}
\label{sec:introduction}
Metabolic networks, formed by a series of metabolic pathways, are
made of intracellular and extracellular reactions
that determine the biochemical properties of a cell by consuming and
producing proteins, and by a set of
interactions that guide and regulate the activity of such
reactions. Cancer, for
example, can sometimes appear in a cell as a result of some pathology
in a metabolic pathway.
These reactions are at the
center of a cell's existence, and are regulated by other proteins,
which can either activate these reactions or inhibit
them.
%
%
These pathways form an intricate
and complex network of chain reactions, and
can be represented in a human
readable form using graphs, called Molecular Interaction Maps (MIMs)
\cite{Kohn:2005,Pommier:2005}
which describe the
cell cycle checkpoint pathways (see for instance Figure~\ref{fig:pommier}).

\begin{figure}
\centering
\includegraphics[width=\linewidth]{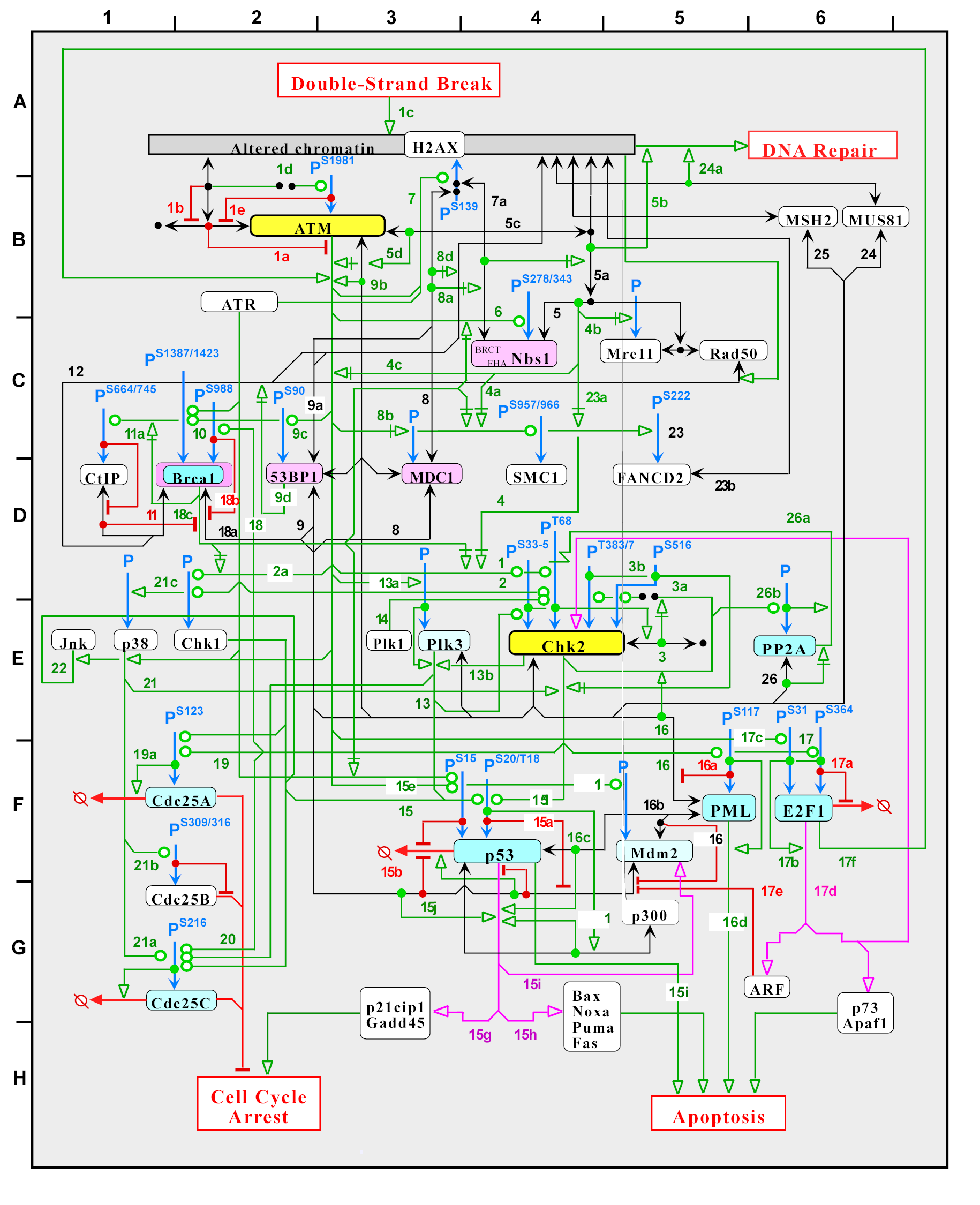}
\caption{atm-chk2/atr-chk1 molecular interaction map.}
\label{fig:pommier}
\end{figure}

Although capital for Knowledge Representation (KR) in biology, MIMs
are difficult to use due to the very large number of elements they may
involve and the intrinsic expertise needed to understand
them.  Moreover, the lack of a formal semantics 
for MIMs makes it difficult to support reasoning tasks commonly
carried out by experts, such as 
checking properties on MIMs, determining how
a MIM can explain a given property or how a MIM can be updated in
order to describe empirically obtained evidences.

This contribution carries on the research undertaken by the authors
aiming at providing a formal background to study MIMs. A first set
of 
works proposed a formalisation of MIMs based on a decidable fragment
of first-order logic \cite{Naji:2013:3,DFO14,DFO14b}. 
\tocheck
In an attempt to find a
simpler representation, without resorting to the expressivity of
first-order logic, other works
\cite{Alliot2016,Alliot2018,ADF16}
 proposed an ad-hoc defined non-monotonic
  logic, called Molecular Interaction Logic (MIL), allowing one to formalize
the notions of  {\em production} and {\em consumption} of
reactives.   In order to formalise the ``temporal evolution'' of a
biological system, MIL formulae are then mapped into Linear Temporal
Logic (LTL) \cite{Pnu77}.

\hide{
\cite{Alliot2018} introduced a formalisation which
includes the notion of {\em production} and {\em consumption} of
reactives.  In order to formalise the ``temporal evolution'' of a
biological system, \cite{Alliot2016,ADF16} proposed an encoding of
MIMs into temporal logic via a translation in terms of temporal logic
programs \cite{AguadoCDPV13}.
}
\hide{
A first prototypal system 
 has been released, able to perform abduction and other tasks 
 on MIMs~\cite{APIA19}, inspired by a pioneer work on
abduction in the field of biological systems~\cite{CP17}. 
}

{This paper embraces the idea, proposed by the above mentioned
  works, that LTL
\tocheck
 is a suitable framework for modelling
  biological systems due to its ability to describe  the
  interaction between components (represented by propositional variables) and
their
  presence/absence in different time instants.
Beyond giving  a formal definition
  of graphs representing MIMs, the paper shows how they can be modeled
  as an LTL theory, by means of a direct ``encoding'', without
  resorting to intermediate (and cumbersome) ad-hoc logics.} 
\hide{
we 
refine and extend the above mentioned works, by 
giving a formal definition
  of graphs representing MIMs and defining how to model them
in Linear Time Temporal Logic  (LTL) \cite{Pnu77}, which appears
to be a suitable framework for reasoning on biological
networks. }
The  logical encoding allows one to formally address 
reasoning tasks, such as, for instance,  checking whether a graph
satisfies some given property $\phi$, as well as  finding out 
which initial conditions would guarantee $\phi$, or else how can the 
the graph be updated 
  in order to satisfy  $\phi$.  
A first prototypal system has been implemented on the
basis of the theoretical work, allowing one to automatically
accomplish reasoning tasks on  MIMs. 

\hide{
Finally, the paper  presents P3M, 
\hide{
a software tool implementing the
encoding method and some interesting reasoning tasks. Grounding
techniques are used in order to reduce the temporal theory to a set of
propositional formulae (assuming bounded time) and therefore exploit 
 classical SAT solvers. 
}
a software platform for modelling, checking
and solving problems on MIMs. This software relies on 
the above introduced methodology and on the 
translation of temporal formulas 
into propositional formulas by grounding, and last by using a SAT
solver.
}


\hide{
 Temporal logic is a suitable framework for modelling
  biological systems due to its ability to describe  the
  interaction between components and 
their
  presence/absence, i.e components are represented by variables and
  interactions are represented by logical formulas. 
}
\hide{
\begin{comment}
The following paragraphs are taken from the IJCAI submission, but they
are probably subsumed by Section \ref{approaches}. To be checked.

Following
  this approach Chabrier et al.~\cite{CC+04} successfully modelled
  biological networks, involving more than 500 genes. In their
  approach, they use temporal logic for checking reachability,
  stability and temporal ordering properties. A similar study on a
  much smaller (although real) biological system has been performed
  in~\cite{BM11}. Also, very famous applications like
  BIOCHAM~\cite{CFS06}, Bio-PEPA~\cite{C08,PP13} and
  ANIMO~\cite{Scholma14} are very popular in the community. We refer
  the reader to~\cite{Sangiovanni,Farinas:14} for an overview on this
  topic.
	
	Regarding MIMs, several
        contributions~\cite{DFO14b,Naji:2013:3,Alliot2016,Alliot2018,ADF16,CP17}
        propose encodings into First-order logic, \emph{Answer Set
          Programming}~\cite{ASP} (ASP) and LTL. The present paper 
        proposes a formal definition of MIMs as  particular automata
        whose associated traces, which can be modelled by 
        LTL formulas, correspond to the biological behaviour of the
        corresponding 
        MIM. 
\end{comment}
}

It is worth pointing out that the adequacy of LTL to model MIMs
  is due to the fact that the latter are qualitative
\tocheck
representations of biological processes. In other terms,
 they model the interactions among the different components
of a biological system without resorting, for instance, to differential
equations like the Systems Biology Markup Language (SBML) \cite{sbml} does.
\hide{
but they do not allow for the specification of
dynamics in an external way, for
instance in terms of differential equations like 
the Systems Biology Markup Language (SBML) \cite{sbml} does.
}

The rest of this paper is organized as follows.  
{Section \ref{approaches} gives a brief overview of modelling
approaches for networks of biological entities.}
Section~\ref{sec:lac}
presents the lac operon that  will be used as a leading  example 
 to introduce all the concepts dealt with by our approach.
 Section~\ref{sec:modelling} describes the fundamental elements and
 concepts of the modelling approach. 
  Section~\ref{sec:mig} presents
Molecular Interaction Graphs (MIGs), which formalize Molecular
Interaction Maps capable of describing and reasoning about general
pathways.  Section~\ref{sec:ltl} explains how MIGs can be represented
by use of Linear Temporal Logic.
Section~\ref{sec:implementation} describes the current state of the
operational implementation of the software
 tool and section~\ref{sec:examples} presents some examples on larger problems.
Finally, Section~\ref{sec:conclusion} concludes this paper and 
discusses possible future work.

\section{Logical Approaches to Biological Systems}
\label{approaches}

The typical objects to be modelled in the framework of systems biology are
\emph{networks} of interacting elements that evolve in time. According
to the features of the network and its properties, various approaches
can be followed,  which can   describe the dynamics of the
system taking the following elements into consideration: 
\begin{itemize}
	\item \textbf{Components}: they are represented by variables,
          which can be either \emph{discrete} or \emph{continuous} depending
          on the requirements of the model. 
	\item \textbf{Interactions}: they are represented by rules
          that specify the dynamical changes in the variables values. These
          interactions can  in their turn be classified according to
 the adopted representation of  time  (\emph{discrete} or
 \emph{continuous}). Finally, the execution of an action can be either
 \emph{stochastic} or not, if a certain degree of
 uncertainty is considered, reflecting the assumption of a noisy environment. 
\end{itemize}

According to the different possible semantics, the various modelling
approaches may be classified as follows~\cite{FH07}:  

\begin{itemize}
	\item Models that involve component quantities and
          deterministic interactions: such models are mathematical,
          inherently quantitative and usually based on ordinary
          differential equations. Tools like Timed Automata
          representations or Continuous-Time Markov Chains are used in
          the construction of models of this category. 
	\item Discrete-value models: they are characterised by the use
          of discrete time. Approaches like executable models based on
          Finite State Machines representations or stochastic models
          such as Discrete-Time Markov Chains belong to this
          category. 
\end{itemize}

\noindent Other hybrid models such as Hybrid Automata or Process
Algebraic Techniques, mix discrete and continuous representation for
both variables and time dynamics. Biological properties
can be distinguished between \emph{qualitative} and \emph{quantitative}:
in the former case, time has an implicit consideration while the
latter involves reasoning on the dynamics of the system along time. To
give an example, \emph{reachability} and  \emph{temporal ordering of
  events} are considered qualitative properties while
\emph{equilibrium states} and \emph{matabolite} dynamics are
quantitative properties.

Gene Regulatory Networks (GRNs) have been very well studied in the
temporal context because the interaction between components may be
easily represented by their presence/absence, i.e components are
represented by  boolean variables and  interactions are represented by
logical rules on their values. Following this approach Chabrier et
al.~\cite{CC+04} successfully modeled a very large network, involving
more than 500 genes. They resorted to  Concurrent Transition Systems
(CTS), allowing one to model modular systems, and
can be then translated into the NuSMV language. They checked 
reachability, stability and temporal ordering properties by the use of CTL. A similar
study on a much smaller (although real) biological system has been
performed in~\cite{BM11}. Here  the LTL specification
syntax and the Spin model checker are used to verify stability properties. 

When a quantitative approach is chosen, the model dimensions drop
drastically. This is essentially due to lack of knowledge on the
parameter values for all the interactions, and 
to the increased computational complexity deriving from a large
model. In this kind of settings, logical approaches have been used to verify temporal
properties on the representations. For instance,
\cite{BRJ+05,Fages04,F08}  use CTL to verify, among other
properties, reachability and stability on different types of
biological networks, and in ~\cite{Brim13} such properties are
checked by using LTL. All these approaches are supported not only for
theoretical results but also for tools and frameworks that allow
biologists to describe a biological network and then verify whether
such representations satisfy some desired properties. Among others,
the systems BIOCHAM~\cite{CFS06}, Bio-PEPA~\cite{C08,PP13} and
ANIMO~\cite{Scholma14} are very popular in the community. 
We refer
  the reader to~\cite{Sangiovanni,Farinas:14} for an overview on this
  topic.

Some
considerations can be made from the study of the aforementioned
contributions: 

\begin{itemize}
	\item the size of the modelled systems is generally very
          small, and a great degree of abstraction and suitable tools
          are needed to deal with large models; 
	\item qualitative approaches are generally enough to analyse a
          large variety of interesting biological properties; 
	\item temporal logic plays an important role in the representation and verification of biological systems.
\end{itemize}

Contrarily to approaches incorporating quantitative information into the
temporal formalisation~\cite{chabrier04}, our contribution belongs to
the category of 
qualitative approaches, since quantitative information in biological
relations, such as the quantity of reactives and their speed of consumption
in a reaction, are not formalised.
{MIMs in fact represent the interaction among the different components
  of the
\tocheck
system and how they evolve in time according to the different reactions. To the best of our knowledge, there is not any contribution where MIMs are used to model quantitative biological information.}


%
\section{A simple example: the lac operon}\label{sec:lac}

\hide{
In this section we revisit the example of the {\em lac}
operon\footnote{The Nobel prize was awarded to Monod, Jacob and Lwoff
  in 1965 partly for the discovery of the lac operon by Monod and
  Jacob~\cite{Monod61}, which was the first genetic regulatory
  mechanism to be understood clearly, and is now a ``standard''
  introductory example in molecular biology classes. See also
  \cite{Lac15}} from~\cite{Alliot2016,ADF16}.
}
	
This section describes a simple example, which represents the
regulation of the  {\em lac} operon (lactose operon),\footnote{The
  Nobel prize was awarded to Monod, Jacob and Lwoff 
  in 1965 partly for the discovery of the lac operon by Monod and
  Jacob~\cite{Monod61}, which was the first genetic regulatory
  mechanism to be understood clearly, and is now a ``standard''
  introductory example in molecular biology classes. See also
  \cite{Lac15}} already used in \cite{Alliot2016,ADF16}.
The lac operon  is an operon required for the transport and
metabolism of lactose in many bacteria. Although glucose is the
preferred carbon source for most bacteria, the lac operon allows for
the effective digestion of lactose when glucose is not available. The
lac operon is a sequence of three genes (lacZ, lacY and lacA) which
encodes 3 enzymes which in turn carry the transformation of lactose
into glucose. We will concentrate here on lacZ which encodes 
$\beta$-galactosidase which cleaves lactose into glucose and
galactose.

The lac operon uses a two-part control mechanism to ensure that the
cell expends energy producing the enzymes encoded by the lac operon
only when necessary. First, in the absence of lactose, the lac repressor
halts production of the enzymes encoded by the lac operon. Second, in the
presence of glucose, the catabolite activator protein (CAP), required
for production of the enzymes, remains inactive.

\begin{figure}[t]
\begin{center}
\includegraphics[scale=0.3]{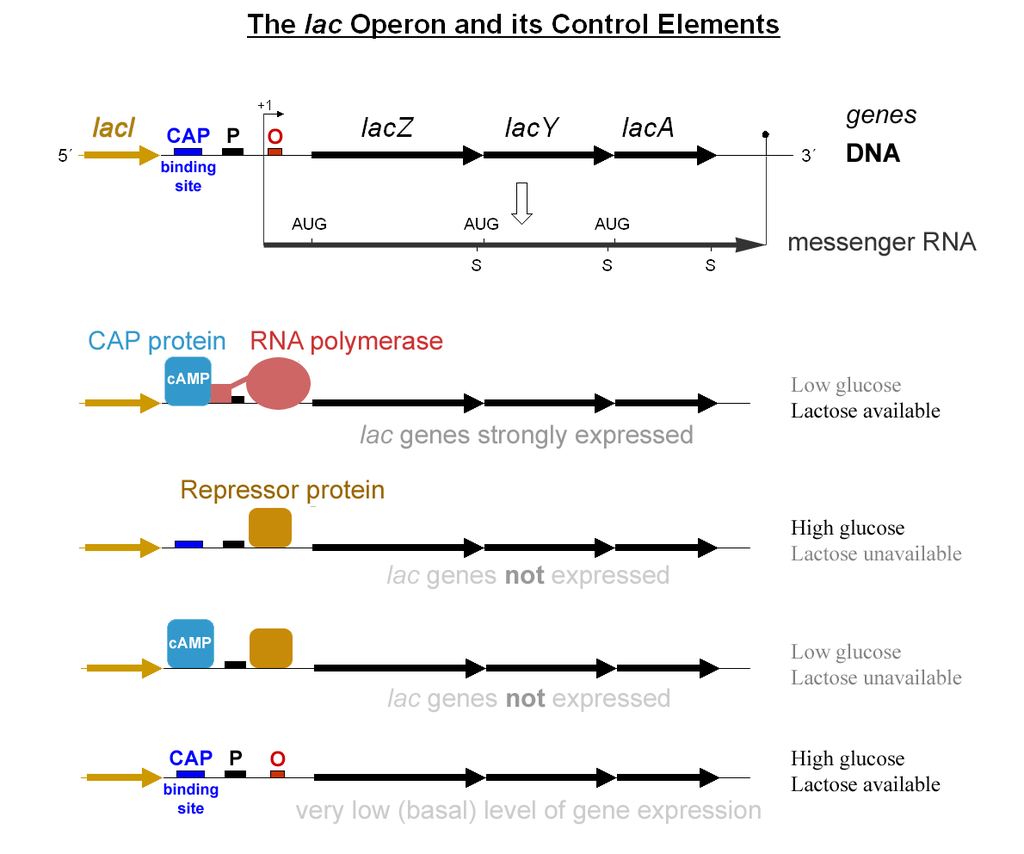}
\end{center}
\caption{The lac operon}
\label{fig:lac}
\end{figure}

Figure~\ref{fig:lac} describes this regulatory mechanism. The
expression of lacZ gene is only possible when RNA polymerase (pink)
can bind to a
promotor site (marked P, black) upstream the gene. This binding is
aided by the cyclic 
adenosine monophosphate (CAMP protein, in blue) which binds before the promotor
on the CAP site (dark blue).

The lacl
gene (yellow) encodes the repressor protein Lacl (yellow) which binds to the
promotor site of the RNA polymerase when lactose is not available,
preventing the RNA polymerase to bind to the promoter and thus
blocking the expression of the 
following genes (lacZ, lacY and lacA): this is a {\em negative regulation},
or {\em inhibition}, as it blocks the production of the proteins.
When lactose is present, one of its isomer, allolactose, binds with
repressor protein Lacl which is no longer able to bind to the promotor site, thus
enabling RNA polymerase to bind to the promotor site and to start
expressing the lacZ gene if CAMP is bound to CAP. 

The CAMP molecule is on the opposite a {\em positive regulation}
molecule, or an {\em activation} molecule, as its presence is necessary
to express the lacZ gene. However, the concentration of CAMP is itself regulated negatively
by glucose: when glucose is present, the concentration of CAMP becomes
low, and thus CAMP does not bind to the CAP site, blocking the
expression of lacZ. Thus glucose prevents the activation by CAMP of the expression of galactosidase from lacZ.

\section{Molecular Interaction Maps (MIMs)}
\label{sec:modelling}

The mechanism described in the previous section is
represented in Figure~\ref{fig:mimlac}, which is an example of
MIM.\footnote{Technically, the generation of CAMP from 
  Adenosine Tri Phosphate (ATP) is blocked by the presence of glucose,
  but we have simplified the graph by simply writing that the presence
  of glucose prevents the activation by CAMP of the expression of
  galactosidase from lacZ.} 
\begin{figure}[t]
\begin{center}\bigskip
\includegraphics[scale=0.4]{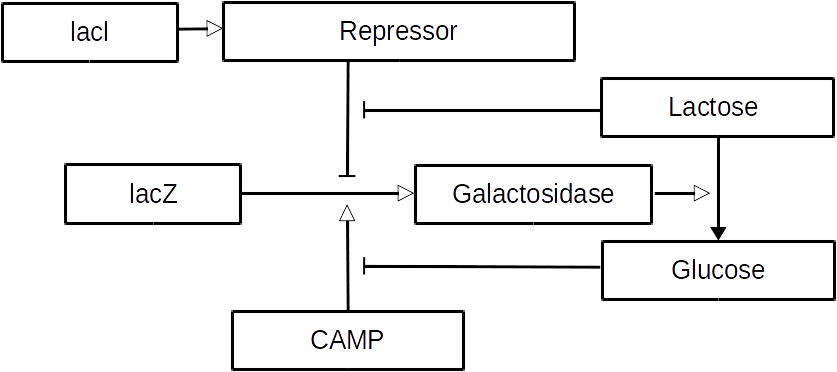}
\end{center}
\caption{Functional representation of the lac operon}
\label{fig:mimlac}
\end{figure}
This example contains all the relations and all the categories of
entities (i.e. the nodes of the graph) 
that we use in our modelling. They are presented below. 

\subsection{Relations}
The relations among the entities
(represented by links in the graphs) represent reactions and can be of
different types:

\begin{description}
\item[Productions] can take two different forms, depending on whether
the reactants are consumed by the reactions or not:

\begin{enumerate}[label=\arabic*)]
\item The graphical notation used 
when a reaction consumes completely the reactant(s) is
  $a_1,\dots,a_n \Aprod b$, meaning that the production of $b$ 
completely consumes $a_1,\dots,a_n$.

 For instance,  in Figure~\ref{fig:mimlac}, 
lactose, when activated by galactosidase produces glucose,
  and is consumed while doing so, which is thus noted by
 $lactose \Aprod glucose$.

\item If the reactants are not completely consumed by the reaction,
  the used notation is 
 $a_1,\dots,a_n \Aactiv b$. Here $b$ is
produced but $a_1,\dots,a_n$ are still present after the
production of $b$.

For example, the
  expression of the 
  lacZ gene to produce galactosidase (or of the lacl gene to produce the
  Lacl repressor protein) does not consume the gene, and we thus have 
  $lacZ \Aactiv galactosidase$.
\end{enumerate}

\item[Regulations] are also of two types: every reaction can be either
{\em inhibited} or {\em activated} by other 
proteins or conditions.
\begin{enumerate}
\item The notation of the type $a_1,\dots,a_n  \Aactiv \dots$ means that the
simultaneous presence of $a_1,\dots,a_n$ activates a production or
another regulation.

 In the example of  Figure \ref{fig:mimlac} the production of galactosidase from the
  expression of the
  lacZ gene is activated by CAMP ($CAMP  \Aactiv(lacZ\Aactiv Galactosidase)$
  expresses activation).

\item The notation $a_1,\dots,a_n \Ainhib\dots$ represents the fact
  that 
simultaneous presence of $a_1,\dots,a_n$ inhibits a production or
another regulation.

In Figure \ref{fig:mimlac}, $Repressor\Ainhib (lacZ\Aactiv Galactosidase)$
represents the fact that production of galactosidase is blocked (or 
  inhibited) by the Lacl repressor protein.
\end{enumerate}

\end{description}

\begin{figure}[t]
\begin{tabular}{cc}
  \begin{minipage}{.45\linewidth}
    \includegraphics[width=\linewidth]{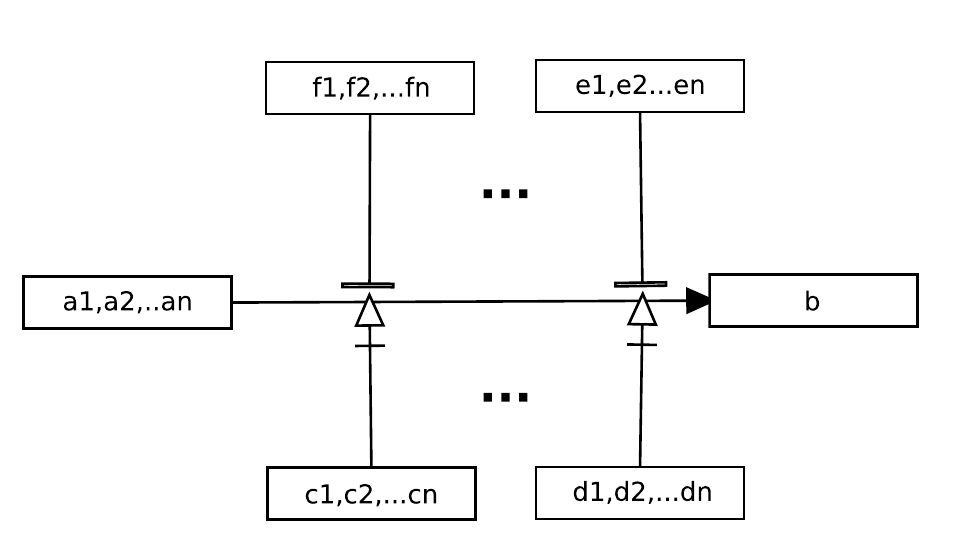}
  \end{minipage}
&
  \begin{minipage}{.45\linewidth}
    \centering
    \includegraphics[width=\linewidth]{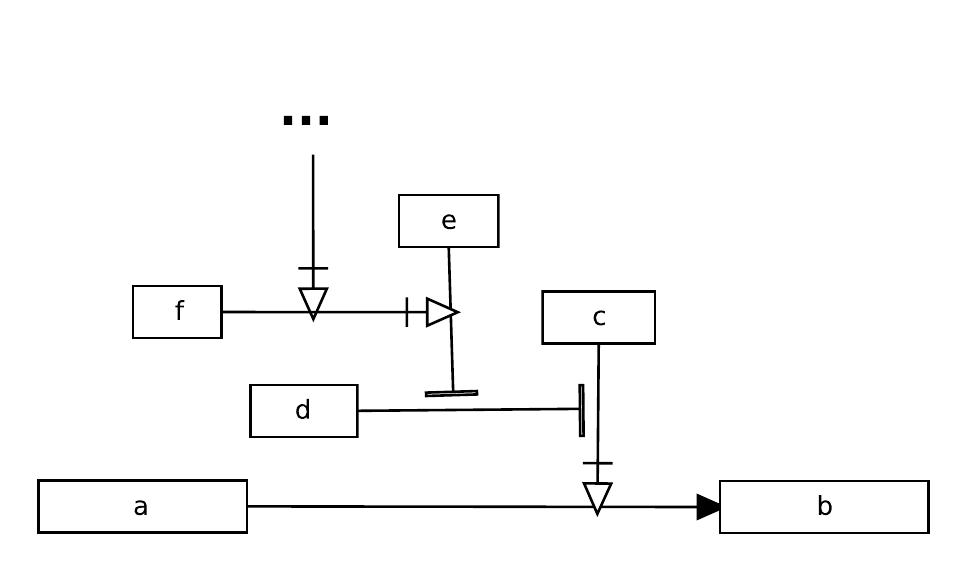}
  \end{minipage}\bigskip
\\
(a) Activations/Inhibitions & (b) Stacking
\end{tabular}
\caption{Activations/Inhibitions and Stacking}
\label{fig:actinh}
\end{figure}

Figure~\ref{fig:actinh}.(a) shows the basic
inhibitions/activations on a reaction: the production of $b$ from
$a_1,\dots,a_n$ is activated by 
 the simultaneous presence of \red{both 
$c_1,\dots,c_n$ {and}  the simultaneous presence of
$d_1,\dots,d_n$}, and inhibited by either the simultaneous presence of
$e_1,\dots,e_n$ {or}  the simultaneous presence of
$f_1,\dots,f_n$.

These regulations are often ``stacked'', on many levels, like shown in
Figure~\ref{fig:actinh}.(b). For example in  
Figure~\ref{fig:mimlac}, the inhibition
by the Lacl repressor protein of the production of galactosidase can
itself be 
inhibited by the presence of lactose, while the activation of the same
production by CAMP is inhibited by the presence of glucose.

\subsection{Types of Entities}
Entities occurring in node labels    can be of  two different types:
\begin{description}
\item[Exogenous:]
the value of an exogenous variable is set once and for all by the
environment or by the experimenter at the start of the simulation and
{\em never} changes through time; 
if the entity is set as present and used in a reaction, the
environment will always provide ``enough'' of it and it will remain
present. 

\item[Endogenous:]
  an endogenous entity can either be  
present or absent at the beginning of the process, as set by the
experimenter, and its value after the start of the process is set
only by the dynamics of the graph.  
\end{description}

These distinctions are fundamental, because the dynamics of 
entities are different and they must be formalized
differently. In practice, the type of an entity is something which is
set by the biologist, according to his professional understanding of the
biological process described by the map. For instance, in
Figure~\ref{fig:mimlac}, the type of the different entities  could be
set as 
follows in order to describe the real behaviour of the lac operon:
lacl, lacZ, CAMP and lactose are initial external conditions of the
model and they do not evolve in time, and are thus exogenous. Note,
  in particular, that lactose can be set as an exogenous entity,
  even if the graph ``says'' that it is consumed when producing
  glucose. Conversely,  galactosidase, the repressor protein and glucose
  can only be produced inside the graph, and are thus endogenous.

It is important to notice that glucose could be set as an exogenous
variable if the experimenter is interested in testing an environment
where glucose is provided externally. Reciprocally, in a more accurate
representation of the lac operon, CAMP would be an endogenous
variable, produced by ATP and regulated by glucose. These graphs are
only a representation and an approximation of the real process,
designed to fit the particular level of description that the
experimenter wants to model.


Although MIMs
 may contain also other kinds of entities or
links, the two kind of entitities and four kinds of interactions
presented above are all
that is needed to 
build the Molecular Interactions Maps we are using in this paper.

\subsection{Temporal evolution}

A MIM can be considered as an automaton which produces sequences of
states of its entities and Linear Temporal Logic formulas can well 
 describe such sequences of states. Time is supposed to be discrete,
and all relations (productions/consumptions) that {\em can} be
executed {\em are} executed simultaneously at each time step.
 An entity
can have two states (or values): absent (0) or present (1). When an
entity is consumed, it becomes absent and when it is produced it
becomes present. {In  other terms, since quantities are not
\tocheck
  taken into account, due also to  the lack of reliable data thereon, 
reactions do not contend to get use of given
  resources: if an entity is present, its quantity is assumed to be
  enough to be used by all reactions needing it.}

This behaviour might look simplistic, as it does not take into account
the kinetic of reactions, {but it reflects a choice underlying
\tocheck
  MIMs representation framework and, as a matter of fact, it is
  nevertheless adequate to handle many problems.}

The software tool that will be described in Section
\ref{sec:implementation}
 provides {default values both for the variables and for
\tocheck
   their classification as exogenous or endogenous}. However, the user
 can modify such default settings through
the graphical interface. 

\hide{
In the software tool that will be described in Section
\ref{sec:implementation}, 
the initial values of the variables must be provided by the user. The
system however provides ``default'' values for them based on some
simple rules; these default values can be modified by the user through
the graphical interface. 
}
\hide{
\subsection{Problems solved}
\label{problems-solved}
The tool can tackle three kinds of problems:
\begin{description}
\item[Graph validation:]
  it consists in checking whether
  ${\cal G}$ satisfies some property expressed in temporal logic and
  it can be performed by doing \emph{model checking}~\cite{HR04} on
  the temporal translation. \uline{a different citation for model checking?}
  
\item[Graph querying:]
  it consists in finding the initial
  conditions that make ${\cal G}$ satisfy some temporal property
  $\varphi$. This operation is equivalent to perform logical \emph{abduction}
  ~\cite{CP17}. \uline{an additional, more general, citation for abduction?}
  
\item[Graph updating:]
  it consists in turning ${\cal G}$ into
  a new graph ${\cal G'}$ that satisfies  some property
  $\varphi$ which is not satisfied by the original ${\cal G}$. It is
  equivalent to apply \emph{model updating}~\cite{DMN06}. 
\end{description}
  }

\section{Molecular Interaction Graphs}\label{sec:mig}

This section is devoted to define
  \emph{Molecular Interaction Graphs} (MIGs), 
 the graph structures 
which are the formal
  representations of MIMs.  The concept of
  \emph{trace}
will also be defined, with the aim of 
  characterising the dynamic behaviour of a MIM.

A MIG is essentially a graph whose vertices are identified with finite
sets of atoms, each of which represents a molecule. %
 Productions are represented by  links connecting
vertices, while regulations are links whose origin is a vertex and whose
target is another link.  
\hide{
The recursive nature of the definition of
edges representing regulations must be managed with some care, in
order to rule out the possibility of having ``circular regulations'',
{i.e.} edges regulating themselves.
}


\begin{definition}[Molecular Interaction Graph]
\label{def:mig}
A Molecular Interaction Graph (MIG) $\G$ is a  tuple  
{$\migtuple$}
where
\begin{itemize}
\item
 $At$ is a \emph{finite} set of
atoms, {partitioned into the sets $Ex$ (the set of {\em exogenous}
  atoms) and $Ed$ (the set of {\em endogenous} ones), 
\item ${\cal B}$ is a set of
  literals from atoms in $At$ (the {\em initial conditions}),} 
\item $\RPr$ and $\RPrc$ are sets of {\em productions}:
\[\begin{array}{lll}
\RPr ~\subseteq~  \{(P\Aactiv Q)\mid P,Q~\subseteq~ At\}&~~&
\RPrc ~\subseteq~  \{(P\Aprod Q)\mid P,Q~\subseteq~ At\}
\end{array}\]
\item $\RAct$ and $\RInh$ are sets of {\em regulations}, such that  for some  $n\in\Nat$:
\[\begin{array}{lll}
\RAct  ~=~  \displaystyle{\bigcup_{i=0}^n \RAct_i}&~~~~~&
\RInh  ~=~  \displaystyle{\bigcup_{i=0}^n \RInh_i}
\end{array}\]
where $\RAct_i$ and $\RInh_i$ are inductively defined  as follows: 
\[\begin{array}{lll}
\RAct_0~\subseteq~ \{(P\Aactiv X )\mid  X  \in\,
\RPr\cup \RPrc\} &~~&
\RInh_0~\subseteq~ \{(P\Ainhib X )\mid  X  \in\,
\RPr\cup \RPrc\}\\
\RAct_{i+1}~\subseteq~ \{(P\Aactiv X )\mid  X  \in\,
\RAct_i\cup \RInh_i\} &&
\RInh_{i+1}~\subseteq~ \{(P\Ainhib X )\mid  X  \in\,
\RAct_i\cup \RInh_i\}
\end{array}\]
\end{itemize}
\noindent 
A {\em  link} is
either a production (i.e.~an element of  $\RPr\cup \RPrc$) or a
regulation (an element of $\RAct\cup \RInh$).
The \emph{depth} of a regulation $X$ is the integer $k$ such that
 $X\in\,\RAct_{k}\cup\RInh_{k}$. 
\end{definition}

{The ``stratified'' definition of regulations rules out the
  possibility
\tocheck
of circular chains of activations and inhibitions.}
Furthermore, it is worth pointing out that, 
 since $At$ is a finite set of atoms, then also the sets
 $\RPr$ and  $\RPrc$ are finite. 
\tocheck
Consequently,  $\RAct$ and
$\RInh$ are finite sets too, since the depth of their elements
is bounded by a given fixed $n\in\Nat$.

{Note that Definition \ref{def:mig} is a generalization
  w.r.t.~the presentation of productions given in Section
  \ref{sec:modelling}, in that it allows for multiple entities on the
  right-hand side of a production. This extension can be considered as
  an abbreviation: $a_1,\dots,a_n \multimap b_1,\dots b_k$ (where
  $\multimap$ is either $\Aactiv$ or $\Aprod$) stands for the set of
  productions $a_1,\dots,a_n \multimap b_1$, \dots,
  $a_1,\dots,a_n \multimap b_k$. 
\label{nota}
}

\begin{example}
\label{example:mig}
The MIG representing the MIM shown in Figure~\ref{fig:mimlac},
{ignoring  the initial conditions and the exogenous and
  endogenous atoms}, 
is constituted by 
\begin{displaymath}
\begin{array}{lll}
At & = &\{Lactose\comma  Galactosidase\comma  Glucose\comma  CAMP\comma lacZ\comma  Galactosidase\comma\\
	  && ~~Lactose\comma  Repressor\comma lacl\}\\
\RPr & = & \{(lacl\Aactiv Repressor),\ (lacZ\Aactiv Galactosidase)\}\\	  
\RPrc &  = & \{(Lactose 
                \Aprod Glucose)\}\\
\RAct & = & \{(CAMP \Aactiv (lacZ\Aactiv Galactosidase)),\\
         && ~~(Galactosidase\Aactiv((Lactose\Aprod Glucose))\}\\
\RInh & =&  \{(Repressor\Ainhib (lacZ\Aactiv Galactosidase)),\\
	  &&~~ (Lactose\Ainhib(Repressor\Ainhib (lacZ\Aactiv Galactosidase))),\\
	  &&~~ (Glucose\Ainhib (CAMP \Aactiv (lacZ\Aactiv
Galactosidase)))\}
\end{array}
\end{displaymath}
\end{example}

{
Having defined the structure of MIGs,
 we now need to provide some
machinery that allows one  to determine the set of substances that trigger
an activation (resp. inhibition) in a MIG. 
The next definition
introduces functions whose values are the regulations directly
activating/inhibiting  a link $X$ in a MIG.  
}

\begin{definition}[Direct regulations of a link]
For every link
$X\in\,\RPr\cup\RPrc\cup\RAct\cup \RInh$:
\[\begin{array}{lll}
\gamma_a(X) & = & \{Y\in\,\RAct\,\,\mid\,  Y  \mbox{ has
  has the form } (P\Aactiv X)\}\\
\gamma_i(X) & = & \{Y \in\,\RInh\,\,\mid\, Y  \mbox{ has 
   the form }(P\Ainhib X)\}
\end{array}\]
\end{definition}

Similarly to transition systems, 
a MIG 
constitutes a compact representation of a set of infinite
sequences of states, where every state {
is determined by the set of
  proteins, genes, enzymes, metabolites, etc that are present in the
  cell at a given time.} 
A sequence of such states thus represents the temporal evolution of
the cell, and will be called a {\em trace}. 
Differently from transition systems, however, the evolution of a MIG
is deterministic: each possible initial configuration determines a
single trace.  {The reason for this is that the representation
  abstract from quantities, hence entities are not considered as
  resources over which reactions may compete (see the remark at the
  end of Section \ref{sec:modelling}).}

Before formally defining the concept of trace, we 
introduce
some preliminary concepts such as the notion of \emph{active} and
\emph{inhibited} links. These two concepts, that are relative to a
given situation ({i.e.} a given set of atoms assumed to be true), 
 will provide the temporal
conditions under which a production can be triggered.

\begin{definition}[Active and inhibited links]
Let  ${\cal G}=\migtuple$ 
be a
MIG. Given $D\subseteq At$ and
$\multimap\in\{\Aactiv,\Aprod,\Ainhib\}$, a link $X=(P\multimap
Y)\in\,\RPr\cup\RPrc\cup\RAct\cup\RInh$ {(where $Y$ is either a set
  of atoms or a link)} is said to be {\em active
  in $D$} if the following conditions hold:

\begin{enumerate}[itemsep=0pt]
\item $P\subseteq D$;
\item every $Z\in\gamma_a(X)$ is active in $D$ -- i.e.
every regulation of the form $Q\Aactiv X\in\, \RAct$
is active in $D$;
\item for all $Z\in\gamma_i(X)$, $Z$  is not active in $D$ -- i.e.
there are no regulations of the form
$(Q\Ainhib X)\in\,\RInh$
that are
  active in $D$.
\end{enumerate}
A link $X\in \, \RPr\cup\RPrc\cup\RAct\cup\RInh$
is {\em
  inhibited in $D$} iff $X$ is not active in $D$.
\end{definition}

\hide{
\begin{comment}
give examples for the mim of lac operon? it would be useful to
understand the example after def. 6
\end{comment}
}

{
Before formalising the concepts of \emph{production} and
\emph{consumption} of substances inside a cell, 
it is worth pointing out that:
\begin{enumerate}[itemsep=0pt,label=\arabic*)]
\item a substance is produced in a cell as a result of a reaction,
  which is triggered whenever the {\em reactants} are present and the
  regulation conditions allow its execution.
\item A substance is consumed in a cell if it acts as a {\em reactive} in a
  reaction which has been triggered.
\item We do not consider quantitative information like concentrations or reaction times: if a
  substance is involved in several reactions at a time, its
  concentration
 does not
  matter,
all reactions will be triggered. Conversely,
  if a substance belongs to the {consumed} {\em reactants}
\tocheck of a 
 triggered reaction,
  it will be completely consumed.
\item I might be the case that
a substance is
  consumed in a reaction while produced by a different one, at the
  same time.
This possibility, that will be further commented below, will however
raise no inconsistency in the definition of traces.

\end{enumerate}
}

\begin{definition}[Produced and consumed atoms]\label{prod-cons-atoms}
Let $\G=\migtuple$
be a MIG and 
$D\subseteq At$. An atom $p\in At$ is {\em produced  in
  $D$} iff {$p\in Ed$ and}
there exists $(P\multimap Q)\in\, \RPr\cup\RPrc$, for 
$\multimap \in \lbrace \Aactiv, \Aprod \rbrace$, such that:
\begin{enumerate}[itemsep=0pt,label=(\roman*)] 
\item $p\in  Q$ and
\item  $(P\multimap Q)$ is active in $D$.
\end{enumerate}

\noindent An atom $p$ is {\em consumed in $D$} iff {$p\in Ed$ and}
there exists
$(P\Aprod Q)\in\, \RPrc$ such that 
\begin{enumerate}[itemsep=0pt,label=(\roman*)]
\item $p\in P$ and
\item $(P\Aprod Q)$ is active in $D$.
\end{enumerate}

\end{definition}

\begin{remark}\label{prod_and_cons}
It may happen that an atom $p$ is both produced
and consumed in a given $D\subseteq At$. Consider, for instance, a MIG with 
$At=Ed=\{p,q,r\}$, 
$\RPr=\{(p\Aactiv q)\}$, $\RPrc=\{(q\Aprod r)\}$,
$\RInh=\RAct=\emptyset$. If $D=\{p,q\}$, the atom
$q$ is produced by $(p\Aactiv q)$ and consumed by $(q\Aprod r)$
in $D$,
since $\{p\}\subseteq D$, $q\in\{q\}\subseteq D$ and 
there are no regulations governing these two  productions.
An even simpler example is given by the (unrealistic) MIG with $At
=Ed=\{p\}$, 
$\RPrc=\{(p\Aprod p)\}$, $\RPr=\RInh=\RAct=\emptyset$ and $D=\{p\}$.
\end{remark}

The behaviour of a MIG 
can be finally formally defined in terms of its
\emph{trace}, taking into account 
activations,
inhibitions, productions and consumptions.

\begin{definition}[Trace] \label{def:trace} 
A \emph{trace} $T$ on a set $At$ of atoms is an infinite
sequence of subsets of $At$, 
$T_0,T_1,\cdots$, called {\em states}. If $\G=\migtuple$
is a MIG,
a \emph{trace for   \G} is a trace $T$ such that:
\begin{enumerate}
\item {$p\in T_0$ for every $p\in{\cal B}$ and
$p\not \in T_0$ for every $\neg p\in{\cal B}$;}
\item 
for all $k\geq 0$ and every atom $p\in At$:
\begin{itemize}
\item {if $p\in Ex$, then $p \in T_{k+1}$ iff $p \in T_{k}$};
\item {if $p\in Ed$, then}
$p \in T_{k+1}$ if and only if either $p$ is  produced in $T_k$
or $p\in T_k$ and $p$  is not consumed in $T_k$.
\end{itemize}
\end{enumerate}
\end{definition}

It is worth pointing out 
that the condition on traces for a given MIG 
${\cal G}$ ensures 
that every change in a state of the trace {affecting endogenous
  atoms} has a justification in
${\cal G}$. Consequently, given the initial state $T_0$ of a trace for
$\G$, all the others are deterministically determined by the
productions and regulations of $\G$. 

As a final observation we remark that, 
when  an atom $p$ is both produced and consumed in a
given $T_k$, production prevails over consumption.  For instance, in a
trace for  the MIG of Remark \ref{prod_and_cons}  
with $T_0=\{p,q\}$, where $q$ is both produced and consumed,
$T_k=\{p,q,r\}$ for all $k\geq 1$.


\section{Representing MIGs in Linear  Temporal Logic}\label{sec:ltl}

This section considers the connection between traces and 
LTL and describes 
how to represent a MIG $\G$ 
 by means of an  LTL theory
   whose models 
are exactly the traces for $\G$.

LTL formulae with only unary future time operators are
 built from
the grammar  
\begin{equation*}
\label{ltl:lang}
\varphi ::= \bot \mid  p \mid \neg \varphi 
  \mid \varphi \vee \varphi
\mid \Next \varphi \mid \Box \varphi
\end{equation*}
where $p$ is an atom  (the other propositional connectives and the ``eventually''
operator can be defined as usual).

An LTL \emph{interpretation} $T$ 
is a trace, {i.e.} an infinite
sequence $T_0,T_1,\dots$ of \emph{states}, where a state is a set of atoms.
The  satisfaction relation $T_k\models\varphi$, where $T_k$ is a state
and $\varphi$ a formula built from a set of atoms $At$, is defined as follows:

\begin{enumerate}
	\itemsep 0 cm
	\item \label{tel:it1} $T_k \models p$ iff $p \in V_k$, for
          any $p \in At$; 
	\item  $T_k \not\models \perp$;
	\item $T_k \models \neg \varphi$ iff $T_k \not \models \varphi$;
	\item \label{tel:it3} $T_k \models \varphi \vee \psi$; 
	iff  $T_k \models \varphi$ or  $T_k \models \psi$;
	\item $T_k \models \Next \varphi$ iff $T_{k+1} \models \varphi$; 
	\item $T_k \models \Box \varphi$ iff for all $j \geq k$,
          $T_j\models \varphi$; 

\end{enumerate}
A formula $\varphi$ is true in an interpretation $T$ if and only if
$T_0\models \varphi$.

A MIG $\G=\migtuple$
is represented by
means of a set of LTL formulae on the set of atoms $At$.
First of all,  classical formulae representing the fact that a given 
link 
is active (or inhibited) are defined.
Below, $\multimap$ stands
for any of 
$\Aactiv$, $\Aprod$ or $\Ainhib$

\begin{definition}\label{def:AI-ltl}
Let $\G=\migtuple$
be a MIG.
If 
$X=
(P\multimap Y)\in\,\RPr\cup\RPrc\cup\RAct\cup\RInh$ {(where $Y$ is
  either a set of atoms or a link)},
then:
\[
\begin{array}{rll}
\acti(P\multimap Y
) & \eqdef  & \displaystyle{\bigwedge_{p\in P} p\wedge 
\bigwedge_{\rho \in\gamma_a(X)} \acti(\rho ) \wedge
\bigwedge_{\rho \in\gamma_i(X)} \inhib(\rho )}
\end{array}\]
{where $\inhib(X)$ is an abbreviation for}
the negation normal form of 
$\neg \acti(X)$.

\hide{
\\
\inhib(X) & =_{def} & \neg \acti(X)
\end{array}\]
\\
\inhib(X) &  =_{def} & \displaystyle{\bigvee_{p\in P} \neg p\vee
\bigvee_{\rho \in\gamma_a(X)} \inhib(\rho ) \vee
\bigvee_{\rho \in\gamma_i(X)} \acti(\rho )}
}

\end{definition}

{It is worth pointing out that both $\acti(X)$ and $\inhib(X)$ are
  classical propositional formulae.}

\begin{example}
Let us consider, for instance, the links of the MIG $\G$ of Example
\ref{example:mig}: 
\[\begin{array}{ll}
1) &(lacl\Aactiv Repressor)\\ 
2) & (lacZ\Aactiv Galactosidase)\\ 
3) & (Lactose
\Aprod Glucose)\\ 
4) & (CAMP \Aactiv (lacZ\Aactiv Galactosidase))\\
5) & (Galactosidase\Aactiv(Lactose\Aprod Glucose)\\
6) & (Repressor\Ainhib (lacZ\Aactiv Galactosidase))\\
7) & (Lactose\Ainhib(Repressor\Ainhib (lacZ\Aactiv Galactosidase)))\\
8) & (Glucose\Ainhib (CAMP \Aactiv (lacZ\Aactiv Galactosidase)))
\end{array}\]
For each of them, $\acti(X)$ and $\inhib(X)$ can be computed as
follows:
\begin{enumerate}[itemsep=0pt]
\item[$\acti(1)$]= $\acti(lacl\Aactiv Repressor)=lacl$;
\item[$\acti(7)$] = $\acti(Lactose\Ainhib(Repressor\Ainhib (lacZ\Aactiv
  Galactosidase)))=Lactose$; 
\item[$\acti(8)$] = $\acti(Glucose\Ainhib (CAMP \Aactiv (lacZ\Aactiv
  Galactosidase)))= Glucose$; 
\item[$\acti(4)$] = $\acti(CAMP \Aactiv (lacZ\Aactiv Galactosidase))= CAMP
  \wedge \inhib(8) =\\ CAMP \wedge \neg Glucose$; 
\item[$\acti(5)$] = $\acti(Galactosidase\Aactiv(Lactose\Aprod Glucose))=
Galactosidase$;
\item[$\acti(3)$] = $\acti(Lactose\Aprod Glucose)=Lactose\wedge \acti(5) =
Lactose\wedge Galactosidase$;
\item[$\acti(6)$] = $\acti(Repressor\Ainhib (lacZ\Aactiv Galactosidase))=
  Repressor \wedge \inhib(7) =\\ Repressor \wedge \neg Lactose$; 
\item[$\acti(2)$] = $\acti(lacZ\Aactiv Galactosidase)= 
lacZ \wedge \acti(4)\wedge
  \inhib(6) =\\lacZ \wedge CAMP \wedge \neg Glucose \wedge (\neg Repressor
  \vee Lactose)$; 
\end{enumerate}
\end{example}

The next result establishes that $\acti(X)$ is an adequate representation of
the property of being active for the link $X$. 
 
\begin{lemma}\label{lemma1}
Let $\G=\migtuple$
be a MIG and
$D\subseteq At$. For every link $X\in\,
\RPr\cup\RPrc\cup\RAct\cup\RInh$: $D\models \acti(X)$ if and only if $X$
is active in $D$.
\end{lemma}

\begin{proof}
Let the {\em size} of a link $X$, $size(X)$, 
 be defined as the number of arrows
$\multimap\,\in\{\Aactiv,\Aprod,\Ainhib\}$ occurring in $X$, and let $M$
be the maximal size of a link in $\G$. If $X=P\multimap Y$ is any
link in $\G$, the proof is by induction on
$k = M-size(X)$.
\begin{itemize}
\item If $k=0$, then $\G$ does not have any link of size
  greater than $size(X)$, hence $\gamma_a(X)=\gamma_i(X)=\emptyset$, 
  $\acti(P\multimap Y) = \bigwedge\limits_{p\in P} P$, and $X$ is
  active in $D$ 
  iff $P\subseteq D$. Clearly, $D\models  \bigwedge\limits_{p\in P} P$ iff
  $P\subseteq D$. 

\item If $k>0$, then, for every 
$Z\in  \gamma_a(X)\cup\gamma_i(X)$, $size(Z)=size(X)+1$,
 hence $M-(k+1)<k$. By the induction
  hypothesis, $D\models \acti(Z)$ iff $Z$ is active in $D$. 
Then the thesis follows from the facts that:
(i) $D\models \bigwedge\limits_{p\in P} p$ iff $ P\subseteq D$; 
(ii) for all 
$Z \in\gamma_a(X)$, $D\models  \acti(Z )$ iff $Z$ is active in
$D$ (by the induction hypothesis), and
(iii) for all $Z \in\gamma_i(X)$, $D\models  \neg \acti(Z )$ iff
$Z$ is
  not active in $D$ (by the induction hypothesis).
\end{itemize}
\end{proof}

\hide{
\begin{proof}
We first prove that the thesis holds when $X=(P\multimap\ell)
\in\, \RAct\cup\RInh$,
by induction on  $n-k$, 
where $n$ is the maximal depth of regulations in $\G$ and
$k$ is the depth of $X$. 
\begin{itemize}
\item If $n-k=0$, then $\gamma_a(X)=\gamma_i(X)=\emptyset$, 
  $\acti(P\multimap\ell) = \bigwedge\limits_{p\in P} P$, and $X$ is active in $D$
  iff $P\subseteq D$. Clearly, $D\models  \bigwedge\limits_{p\in P} P$ iff
  $P\subseteq D$. 
\item Let $j=n-k>0$. For every $\rho\in \gamma_a(X)\cup\gamma_i(X)$, the
  depth of $\rho$ is $k+1$, hence $n-(k+1)<j$. By the induction
  hypothesis, $D\models \acti(\rho)$ iff $\rho$ is active in $D$. 
Then the thesis follows from the facts that:
(i) $D\models \bigwedge\limits_{p\in P} p$ iff $ P\subseteq D$; 
(ii) for all 
$\rho \in\gamma_a(X)$, $D\models  \acti(\rho )$ iff $\rho$ is active in
$D$ (by the induction hypothesis), and
(iii) for all $\rho \in\gamma_i(X)$, $D\models  \neg \acti(\rho )$ iff
$\rho$ is
  not active in $D$ (by the induction hypothesis).
\end{itemize}
When $X\in\,\RPr\cup\RPrc$, the thesis is proved similarly to
the induction step above, 
just exploiting the fact that the elements in 
$\gamma_a(X)$ and $\gamma_i(X)$ are regulations, for which the
required equivalence holds.
\end{proof}
}

{In order to give a more compact presentation of the LTL theory
  representing a MIG, we
  define, for each atom $p\in At$,  classical formulae representing
  the fact that $p$ is produced or consumed.}

{
\begin{definition}
\label{defi:Pr-Cn}
Let Let $\G=\migtuple$
be a MIG,
 ${\cal P}rod=\RPr\cup\RPrc$, and $p\in At$. 
Then:
\[
\begin{array}{lll}
\Prod{p} &\eqdef&\left\{ 
        \begin{array}{ll}
        \bot & \mbox{ if }p\in Ex\medskip\\
        \bigvee_{(P\multimap Q)\in{\cal P}rod,\,p\in Q}\acti(P\multimap Q)  
             & \mbox{ if }p\in Ed
        \end{array}\right.\\
\\
\Con{p}&\eqdef&\left\{ 
     \begin{array}{ll}
        \bot & \mbox{~~if }p\in Ex\medskip\\
 \bigvee_{(P\rightarrowTriangle Q)\in\RPrc,\,p\in P}\acti(P\rightarrowTriangle Q)
     & \mbox{~~if }p\in Ed
        \end{array}\right.
\end{array}\]
\end{definition}
} 

{
\begin{example} \label{example:encoding}

Let us consider the simple 
MIG $\G$ of Example 	\ref{example:mig}, 
where atoms are partitioned into
			$Ex=\{lacl,lacZ,CAMP\}$ and 
$Ed=\{Repressor\comma 
Lactose\comma Galactosidase\comma Glu$ $cose\}$.\footnote{{In this
example we assume that lactose is endogenous, because it is the only
					consumed entity in the simple
                                        MIM of figure
                                        \ref{fig:mimlac}.}}
The abbreviations $\Prod{p}$ and $\Con{p}$ for the endogenous atoms
		are the following:
		\[\begin{array}{lllll}
		\Prod{Repressor} & \eqdef & \acti(lacl\Aactiv Repressor) \\&
		\eqdef& lacl\\
		\Prod{Lactose} & \eqdef & \bot\\
		\Prod{Galactosidase} & \eqdef & \acti(lacZ\Aactiv Galactosidase)\\ &
		\eqdef& 
		\multicolumn{3}{l}{lacZ \wedge CAMP \wedge\neg Glucose}\\
                &&\multicolumn{3}{l}{~~~ \wedge (\neg Repressor
			\vee Lactose)}\\
		\Prod{Glucose} & \eqdef & \acti(Lactose\Aprod Glucose) \\
		&\eqdef& Lactose\wedge Galactosidase\\
		\Con{Lactose} &\eqdef & \acti((Lactose\Aprod Glucose) \\&
		\eqdef& Lactose\wedge Galactosidase\\
		\Con{p} &\eqdef & 
		\multicolumn{3}{l}{\bot \mbox{ for }p\in \{Repressor,Galactosidase,Glucose\}}
		\end{array}
		\]
	\end{example}} 

Finally, the set of LTL formulae ruling the overall behaviour of  a
MIG  can be defined.

{
\begin{definition}\label{def:encoding}
If $\G=\migtuple$ 
is a MIG,
the \emph{LTL encoding of  $\G$} is the
  set of  formulae containing {all the literals in ${\cal B}$} and,
  for every $p\in At$, the formula 
\[\Box(\Next p \leftrightarrow \Prod{p}\vee (p
 \wedge \neg \Con{p}))\]

\end{definition}
}
\medskip

{It is worth pointing out that, if $p\in Ex$, then the 
  formula encoding its behaviour is equivalent to $\Box(\Next
  p\leftrightarrow
  p)$.
For endogenous atoms, the encoding 
captures the (negative and positive) effects produced by a
reaction on the environment at any time. This encoding has some
similarities with the \emph{successor state axioms} of the \emph{Situation
  Calculus}~
\cite{Reiter2001}.
}

\begin{example}
If $\G$ is the MIG of Example~\ref{example:mig}, 
		the LTL encoding of $\G$  contains (formulae equivalent to)
		$ \Box(\Next lacl\leftrightarrow lacl)$, and 
		similar ones for $lacZ$ and $CAMP$.

		Furthermore, it contains the following formulae,
		ruling the behaviour
		of endogenous atoms:
		\[\begin{array}{ll}
		\multicolumn{2}{l}
		{\Box(\Next Repressor \leftrightarrow \Prod{Repressor}\vee (Repressor
			\wedge \neg \Con{Repressor}))}\\
		~~~~~~~~~&\equiv \Box(\Next Repressor \leftrightarrow lacl \vee
		Repressor)\\
		\multicolumn{2}{l}
		{\Box(\Next Lactose \leftrightarrow \Prod{Lactose}\vee (Lactose
			\wedge \neg \Con{Lactose}))}\\
		~~~~~~~~~&\equiv \Box(\Next Lactose \leftrightarrow 
		Lactose \wedge \neg (Lactose\wedge Galactosidase))\\
		\multicolumn{2}{l}
		{\Box(\Next Galactosidase \leftrightarrow
                  \Prod{Galactosidase}\vee} \\
		&\multicolumn{1}{c}{ (Galactosidase	\wedge \neg
                  \Con{Galactosidase}))}\\ 
		~~~~~~~~~&\equiv \Box(\Next Galactosidase \leftrightarrow 
		(lacZ \wedge CAMP \wedge \neg Glucose \\
		&~~~~~~~~~~~~~~~\wedge (\neg Repressor
		\vee Lactose) \vee Galactosidase))\\
		\multicolumn{2}{l}
		{\Box(\Next Glucose \leftrightarrow \Prod{Glucose}\vee (Glucose
			\wedge \neg \Con{Glucose}))}\\
		~~~~~~~~~&\equiv \Box(\Next Glucose \leftrightarrow 
		(Lactose\wedge Galactosidase)\vee Glucose)
		\end{array}
		\]

\end{example}



The rest of this section is devoted to show that 
the LTL encoding of a MIG  correctly and
completely represents its behaviour. First of all, we prove that 
the {truth} of $\Prod{p}$ and $\Con{p}$ in a state coincide with the
atom $p$ being produced/consumed at that state.

\begin{lemma}\label{lemma2}
If $T$ is a model of the LTL encoding of a MIG,
then for every $k$ and every atom
$p\in {At}$, $p$ is produced in $T_k$ iff {$T_k\models \Prod{p}$} and 
$p$ is consumed in $T_k$ iff  {$T_k\models\Con{p}$}.
\end{lemma}

\begin{proof}
Let $T$ be a model of $\G=\migtuple$
${\cal P}rod=\RPr\cup\RPrc$, 
$k\in\Nat$ and $p\in At$.
\begin{enumerate}
\item If  $T_k\models \Prod{p}$ then 
$\Prod{p}\neq\bot$
and 
there exists some $(P\multimap Q)\in {\cal P}rod$ such that $p\in Q$ and
$T_k\models \acti(P\multimap Q)$. By Lemma
  \ref{lemma1}, $(P\multimap Q)$ is active in $T_k$.
Moreover, since $\Prod{p}\neq\bot$,  $p\in Ed$.
Therefore, from Definition \ref{prod-cons-atoms} it follows that $p$
is produced 
in $T_k$.

\item If  $p$ is produced
in $T_k$, then $p\in Ed$ and there exists some $(P\multimap
Q)\in\,{\cal P}rod$, 
such that $p\in  Q$ and $(P\multimap Q)$ is active in $T_k$.
By Lemma \ref{lemma1}, $T_k\models \acti(P\multimap Q)$, 
hence $T_k\models \Prod{p}$ by Definition \ref{defi:Pr-Cn}, since $p\in Ed$.
\item If  $T_k\models \Con{p}$ then 
$\Con{p}\neq \bot$ and there exists some
$(P\Aprod Q)\in {\cal C}$ such that $p\in P$ and
$T_k\models \acti(P\Aprod Q)$,
By Lemma
  \ref{lemma1}, $P\Aprod Q$
is active in $T_k$.
Moreover, since $\Con{p}\neq\bot$,  $p\in Ed$.
Therefore,
from Definition \ref{prod-cons-atoms} it follows that $p$ is consumed
in $T_k$.

\item If  $p$ is consumed
in $T_k$, then $p\in Ed$ and there exists some 
$(P\Aprod Q)\in {\cal C}$
such that $p\in  P$ and 
$(P\Aprod Q)$ is active in $T_k$.
By Lemma \ref{lemma1}, 
$T_k\models \acti(P\Aprod Q)$, therefore 
$T_k\models \Con{p}$ by Definition \ref{defi:Pr-Cn}, since $p\in Ed$.
\end{enumerate}
\end{proof}

\hide{
If $\G$ is a MIG with sets of atoms $At$, the language of its LTL
axiomatization is an extension of $At$. Consequently, in order to
relate traces for $\G$ and models of its encoding, 
the
notion of trace for $\G$ must be extended to traces on
the extended language $At'=At\cup \{\Prod{p},\Con{p}\mid p\in 
\red{Ed}\}$.


\begin{definition}[Extended trace for a MIG]\label{def:extended-trace}
\label{extended-trace}
Let $\G$ be a MIG and $T=T_0,T_1,\dots$ 
a trace for $\G$. The extended trace 
obtained from $T$ is the trace $T'=T'_0,T'_1,\dots$ such that for
every $k\in\Nat$: 
\begin{enumerate}[itemsep=0pt,label=\arabic*)]
\item\label{extrace:cond1} 
$T'_k\subseteq At'=At\cup \{\Prod{p},\Con{p}\mid p\in \red{Ed}\}$;
\item\label{extrace:cond2} 
for all $p\in At$, $p\in T'_k$ iff $p\in T_k$;
\item\label{extrace:cond3} 
for all $p\in \red{Ed}$, $\Prod{p}\in T'_k$ iff $p$ is produced in $T_k$;
\item\label{extrace:cond4} 
for all $p\in \red{Ed}$, $\Con{p}\in T'_k$ iff $p$ is consumed in $T_k$.
\end{enumerate}

An extended trace for $\G$ is an extended trace obtained from a trace for $\G$.
\end{definition}

It is immediate to see that every extended trace for $\G$ is a trace
for $\G$. The following theorem demonstrates that our LTL encoding of
a MIG correctly and completely represents its behaviour. 
}

\noindent
The adequacy of the LTL encoding of a MIG  can finally
be proved.

\begin{theorem}[Main result]\label{th:main}
If $\G$ is a MIG, 
then:
\begin{enumerate}
\item every  trace for  $\G$  is a model of the LTL encoding of $\G$;
\item every model of the LTL encoding of $\G$ is a trace for $\G$.
\end{enumerate}
\end{theorem}

\begin{proof}
Let us assume that $T$ is a trace for $\G=\migtuple$. 
{Clearly, for every literal 
$\ell\in{\cal B}$, $T_0\models \ell$, since $\ell$ belongs to the
  encoding of $\G$.}
Moreover, 
for all $k\geq 0$ and every atom $p\in At$:
\begin{itemize}
\item if $p\in Ex$, then $p\in T_{k+1}$ if and only if $p\in T_k$.
Hence, $T_k\models \Next p\leftrightarrow p$,   {i.e.} 
$T_k\models \Next p \leftrightarrow \Prod{p}\vee (p \wedge \neg \Con{p})$.
\item If $p\in Ed$, then 
$p \in T_{k+1}$ if and
only if either $p$ is  produced in $T_k$ or $p\in T_k$ and $p$  is not
consumed in $T_k$.  By Lemma \ref{lemma2}, this amounts to saying that
$p \in T_{k+1}$ if and
only if either $T_k\models \Prod{p}$ or $p\in T_k$ and $T_k\not\models
\Con{p}$.  Consequently,
$T_k\models \Next p \leftrightarrow \Prod{p}\vee (p \wedge \neg
\Con{p})$.
\end{itemize}
Since these properties  hold for all
$k$, it follows that for all $p\in At$, 
$T\models \Box(\Next p \leftrightarrow \Prod{p}\vee (p \wedge \neg \Con{p}))$.

For the other direction, let us assume that $T$ is a model of the LTL
encoding of $\G$. Then, in particular, $T_0\models {\cal B}$, hence 
$p\in T_0$ for
every $p\in{\cal B}$, and $p\not\in T_0$ for
every $\neg p\in{\cal B}$. 
Moreover, for all $k\geq 0$ and every atom $p\in At$:
\begin{itemize}
\item  if $p\in Ex$, then $T_k\models \Next p\leftrightarrow p$,
 hence  $p\in T_{k+1}$ if and
  only if $p\in T_k$. 
\item If $p\in Ed$, then 
$T_k\models \Next p \leftrightarrow
   \Prod{p}\vee (p \wedge \neg \Con{p})$, hence 
$p \in T_{k+1}$ if and
only if either $T_k\models \Prod{p}$ or $p\in T_k$ and $T_k\not\models
\Con{p}$. By lemma \ref{lemma2}, this amounts to saying that
$p \in T_{k+1}$ if and
only if either $p$ is  produced in $T_k$ or $p\in T_k$ and $p$  is not
consumed in $T_k$. 
\end{itemize}
Consequently, $T$ is a trace for $\G$.
\end{proof}

\section{Bounding Time and Reduction to SAT}
\label{sec:sat}

\tocheck
The use of an LTL formalization allows us to consider solutions with
infinite length when performing reasoning tasks such as abduction or
satisfiability checks. However, LTL tools for abduction are 
not as developed as in the case of propositional logic, since
\tocheck
the abductive task is in general very complex.\footnote{A method to
  perform abduction for a  fragment of LTL sufficient to represent
  problems on MIMs has been proposed
  in \cite{CP17}, but it has not been implemented.}
In order to take advantage of the highly efficient tools for propositional reasoning
such as SAT-solvers, abduction algorithms, etc, 
the solver that will be presented in Section \ref{sec:implementation}
reduces the problem to propositional logic by
assuming bounded time.
In essence, the  reduction   simulates the truth value of an LTL propositional
variable $p$ along time by a finite set of $n$ fresh atoms, one per time instant. Moreover, the
behaviour of the ``always'' temporal operator is approximated by use of finite conjunctions.
Exogenous variables are not grounded, since it is useless and expensive to consider
different variables in this case.


{
In detail, 
 the grounding to a given time $k\in\Nat$ of a
propositional formula $\varphi$ built from a set of atoms partitioned into
exogenous and endogenous is first of all defined.

\begin{definition}[Grounding of propositional formulae]
	\label{def:prop-grounding}
	Let $\varphi$ be a propositional formula built from the set of atoms
	$At=Ex\,\dot\cup\,Ed$.  The grounding of $\varphi$
	to time $k$, $\inftr{\varphi}_k$, is defined as follows:
	\begin{itemize}
		\item if $p\in Ex$, then $\inftr{p}_k \eqdef p$;
		\item if $p\in Ed$, then 
		$\inftr{p}_k \eqdef p_k$, where  $p_k$ is a
		new propositional variable; 
		\item $\inftr{\neg \varphi}_k \eqdef \neg \inftr{\varphi}_{k}$;
		\item $\inftr{\varphi \vee \psi}_k \eqdef \inftr{\varphi}_k \vee
		\inftr{\psi}_k$.  
	\end{itemize}
	If $S$ is a set of proposional formulae, then 
	$\inftr{S}_k=\{\inftr{\varphi}_k\mid \varphi\in S\}$.
\end{definition}

Next, the grounding of the encoding of a MIG is defined.

\begin{definition}[Grounding of the encoding of a MIG]
	\label{def:grounding}
	Let $\G=\migtuple$ be a MIG, $S$ its LTL encoding and $k\in\Nat$.
	
	For all  $p\in Ed$, if  $SSA_p$ is the formula $\Box(\Next p
	\leftrightarrow \Prod{p}\vee (p  \wedge \neg \Con{p}))$ belonging to
	$S$, we define
	\[\inftr{SSA_p}_k=p_{k+1} \leftrightarrow \inftr{\Prod{p}\vee (p
		\wedge \neg \Con{p})}_k\]
	
	The grounding $\inftr{S}_k$ of $S$ up to time $k$ is defined as follows:
	\[\inftr{S}_k=\{ \inftr{\ell}_0\mid \ell\in {\cal B}\}\cup
	\{\inftr{SSA_p}_i\mid p\in Ed \mbox{ and } 0\leq i<k\}\]
\end{definition}

The grounding  $\inftr{SSA_p}_k$ is well defined, since 
$\Prod{p}\vee (p\wedge \neg \Con{p})$ is a classical formula.
Note that ``successor state axioms'' $SSA_p$
in the LTL encoding of $\G$
are grounded only 
for endogenous variables and only as far as the ``$\Next p$'' refers
to a state that ``exists'' in the bounded timed model.

The next definition formalizes the notion of a temporal interpretation
$T$ and a classical one $M$ being models of the same initial state.
\hide{
In order to prove a kind of ``model correspondence'' property, 
the meaning of 
correspondence between temporal and propositional interpretations has
to be defined.
}

\begin{definition}
	Let $At=Ex\,\dot\cup\,Ed$ be a set of atoms, 
	$T=T_0,T_1,\dots$  an LTL interpretation of the language $At$ and
	$k\in\Nat$. 
	A
	classical interpretation $M$ is said to correspond to $T$ up to time
	limit $k$ if $M$ is an interpretation of the  language 
	$Ex\cup\{p_i\mid p\in Ed\mbox{ and }0\leq i\leq k\}$ and
	for all $p\in At$, $M\models \inftr{p}_0$ iff $T_0\models p$.
\end{definition}


The next result establishes a kind of 
``model correspondence'' property.

\begin{theorem}[Model correspondence]
	Let $\G=\migtuple$ be a MIG, $S$ its LTL encoding, and
	$\inftr{S}_n$  the grounding of $S$ up to time $n$.
	If $T=T_0,T_1,\dots$ is any model of $S$ and 
	$M$ a model of $\inftr{S}_n$ corresponding to $T$, then 
	for every classical propositional formula $\varphi$ and every
	$k=0,\dots,n$:  $M\models
	\inftr{\varphi}_k$ iff $T_k\models\varphi$.
\end{theorem}

\begin{proof}
	By double induction on $k$ and $\varphi$.
	\begin{enumerate} 
		\item If  $k=0$, the thesis is proved by induction on $\varphi$.
		\begin{enumerate}
			\item 
			If $\varphi$ is an atom, then the thesis follows
			immediately from  the fact that $M$ corresponds to $T$. 
			\item\label{case2} If $\varphi=\neg\varphi_0$ 
			or $\varphi=\varphi_0\vee \varphi_1$, the thesis follows
			from the induction hypothesis, the definition of $\inftr{\varphi}_k$
			(Definition \ref{def:prop-grounding}) 
			and the  definition of
			$\models$ for classical logic.
		\end{enumerate}
		
		\item $0<k\leq n$: By the induction hypothesis 
		$T_{k-1}\models\varphi$ iff $M\models \inftr{\varphi}_{k-1}$ for
		every propositional formula $\varphi$.
		The thesis is proved by induction on $\varphi$:
		\begin{enumerate}
			\item If $\varphi$ is an atom, we consider two cases:
			\begin{enumerate}
				\item $p\in Ex$: since $T_{k-1}\models \Next p \leftrightarrow p$, 
				then $T_{k}\models p$ iff $T_{k-1}\models p$. By the induction
				hypothesis, $T_{k-1}\models p$ iff $M\models \inftr{p}_{k-1}$. Since
				$ \inftr{p}_{k-1} = p = \inftr{p}_{k}$, $T_{k}\models p$ iff $M\models
				\inftr{p}_{k}$. 
				\item $p\in Ed$: since  $T_{k-1}\models \Next p \leftrightarrow
				\Prod{p}\vee (p \wedge \neg \Con{p})$, $T_k\models p$ iff 
				$T_{k-1}\models \Prod{p}\vee (p \wedge \neg \Con{p})$. By the
				induction hypothesis, the latter assertion holds iff
				$M\models \inftr{\Prod{p}\vee (p \wedge \neg \Con{p})}_{k-1}$.
				By Definition \ref{def:grounding}, $\inftr{S}_n$ contains
				$p_k\leftrightarrow \inftr{\Prod{p}\vee (p \wedge \neg
					\Con{p})}_{k-1}$, and, since $M\models \inftr{S}_n$, 
				$M\models \inftr{\Prod{p}\vee (p \wedge \neg \Con{p})}_{k-1}$ iff
				$M\models p_k$. Therefore, $T_k\models p$ iff $M\models p_k$.
			\end{enumerate}
			\item If $\varphi=\neg\varphi_0$ 
			or $\varphi=\varphi_0\vee \varphi_1$, the thesis follows
			from the induction hypothesis, 
			Definition \ref{def:prop-grounding}
			and the  definition of
			$\models$ for classical logic, like in the
                        base case.
		\end{enumerate}
	\end{enumerate}
\end{proof}

The rest of this section is devoted to establish the complexity of
grounding for the encoding of a MIG.
%
Let the size of a formula be measured in terms of the number of its logical
operators:
if $\varphi$ is a formula, $\size{\varphi}$
is the number of logical operators in $\varphi$. If $S$ is a set of
formulae, then $\displaystyle{\size{S}=\sum_{\varphi\in
    S}\size{\varphi}}$.

\begin{theorem}[Complexity of the encoding]
\label{encoding-complexity}
	Let $\G$ be a MIG, $S$ its LTL encoding and $\inftr{S}_n$ the grounding
	of $S$ up to time $n$. Then 
	$\size{\inftr{S}_n}\leq n\times
	\size{S}$.
\end{theorem}

\begin{proof}
	First of all we note that if $\varphi$ is a classical formula, then
	$\size{\varphi}= \size{\inftr{\varphi}_k}$ for any $k$.
	Consequently, 
\[\size{p_k\leftrightarrow \inftr{\Prod{p}\vee (p \wedge \neg
			\Con{p})}_{k-1}}= 
	\size{\Next p \leftrightarrow
		\Prod{p}\vee (p \wedge \neg \Con{p})} -1\]
	and  $\size{\inftr{SSA_p}_k}=\size{SSA_p}-2$.
	
	Let $S$ be the LTL encoding of a MIG $\G=\migtuple$
	and $\inftr{S}_n$ its grounding up to time $n$.
	\begin{enumerate}
		\item For each $\inftr{\ell}_0\in \inftr{S}_ n $ such that
		$\ell\in{\cal B}$, $\size{\inftr{\ell}_0} =\size{\ell}$. Therefore
		$\size{\inftr{{\cal B}}_0}=\size{\cal B}$. 
\hide{
Clearly, moreover, the time
			required to generate $\inftr{{\cal B}}_0$ is linear in the size of 
			$\cal B$, since every $\inftr{\ell}_0$ can be generated in constant
			time.}
		
		\item Beyond the literals in $\inftr{{\cal B}}_0$, $\inftr{S}_ n $
		contains $\inftr{SSA_p}_k$ for all $p\in Ed$ and $0\leq k<
		 n $. 
		Hence, for every $SSA_p\in S$, $\inftr{S}_ n $ contains $ n -1$
		formulae, the size of each of them being smaller than the size of
		$SSA_p$. Therefore
		\[\size{\{\inftr{SSA_p}_k \mid p\in Ed\mbox{ and }0\leq k< n \}}
		< n \times \size{\{SSA_p\mid p\in Ed\}}\]
		
		\hide{Moreover, the time needed to generate a single $\inftr{SSA_p}_k$
			is constant w.r.t. the size of $SSA_p$, therefore the time needed to
			generate the set 
			$\{\inftr{SSA_p}_k \mid p\in Ed\mbox{ and }0\leq k< n \}$ is
			linear in the size of $\{SSA_p\mid p\in Ed\}$.
		}
	\end{enumerate}
	Therefore, $||\inftr{S}_ n ||\leq  n  \times ||S||$. 
\end{proof}


It is worth pointing out that exogenous variables are not grounded.
\hide{
 in other terms, 
 it is useless and expensive to consider different variables
$p_0,p_1,\dots$, since a single $p$ is enough.
Hence, when $p$ is exogenous, in the above defined translation,
 $p_i$ is to be replaced by $p$, for all $i$.
}
Consequently, for instance, if $Lactose$ is assumed to be exogenous,
the grounding up to  time $k$ of the 
LTL formula 
$ \Box(\Next Glucose \leftrightarrow 
(Lactose\wedge Galactosidase)\vee Glucose)$ 
is the conjunction of all the formulae of the form
$Glucose_{i+1} \leftrightarrow 
(Lactose\wedge Galactosidase_i)\vee Glucose_i$
for $0\leq i< k$.}

\hide{
$\Box ( Lactose\wedge Galactosidase\imp
\Prod{Glucose})$  is the conjunction of all the formulae of the form
$Lactose\wedge Galactosidase_i \imp \Prod{Glucose}_i$, 
for $0\leq i\leq k$.
}

\section{The P3M tool: a software platform for modelling and
  manipulating MIMs
   } 
\label{sec:implementation}

In this section we present P3M ({\bf P}latform for {\bf M}anipulating
{\bf M}olecular Interaction {\bf M}aps), a prototypal system implementing the
representation mechanism outlined in the previous sections and able to
solve the following problems, that will be discussed further on: graph
validation, graph querying and graph updating.  The system is written
in Objective Caml \cite{ocaml}, and interfaces with the C
implementation of the Picosat solver library~\cite{Picosat08}. A
graphical user interface has been developed to help biologists to
interact with the system in a user-friendly way.  The general
architecture of the system is represented in Figure \ref{fig:implant},
and will be further explained below.
\red{P3M can be downloaded at  \url{http://www.alliot.fr/P3M/}.}

\begin{figure}
\centering
\includegraphics[width=.4\textwidth]{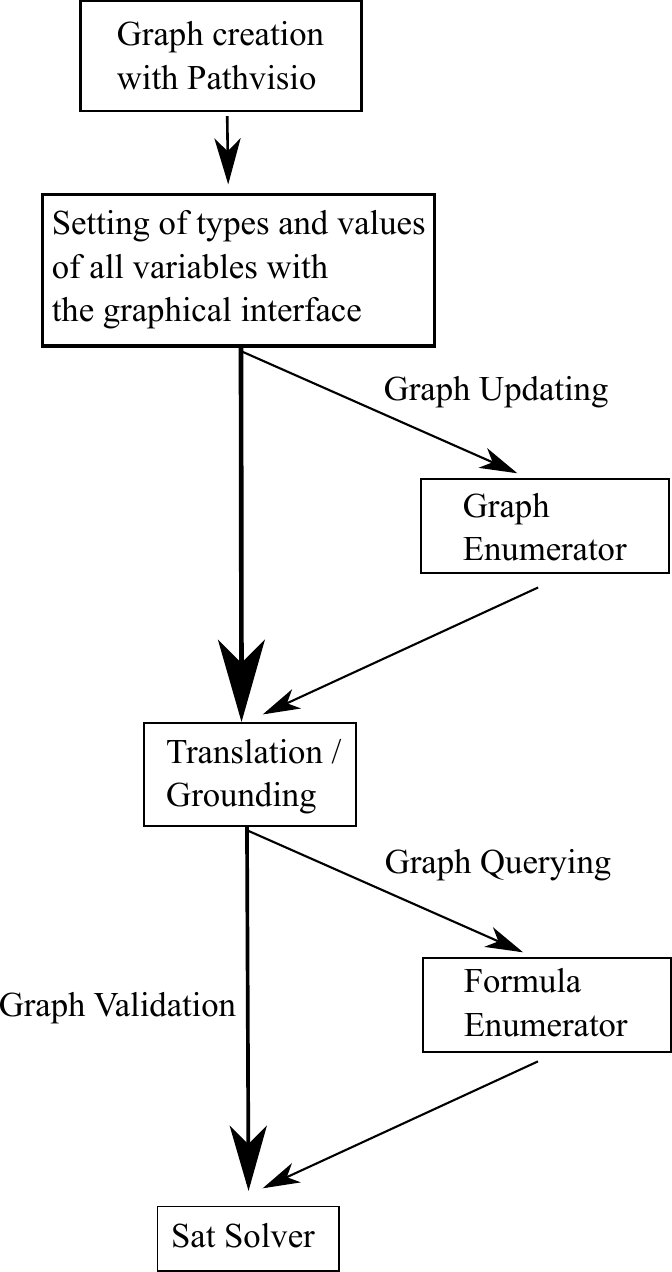}
\caption{Implementation}
\label{fig:implant}
\end{figure}

\hide{
\begin{description}
\item[Graph validation:]
  it consists in checking whether
  ${\cal G}$ satisfies some property expressed in temporal logic and
  it can be performed by doing \emph{model checking}~\cite{CGP99,QS82,VWS15,CGJ+03,HR04} on 
  the temporal translation.
  
\item[Graph querying:]
  it consists in finding the initial
  conditions that make ${\cal G}$ satisfy some temporal property
  $\varphi$. This operation is equivalent to perform logical \emph{abduction}
  ~\cite{CP93,CP95,EG95,Inoue92,DF91,CP17}. 
  
\item[Graph updating:]
  it consists in turning ${\cal G}$ into
  a new graph ${\cal G'}$ that satisfies  some property
  $\varphi$ which is not satisfied by the original ${\cal G}$. It is
  equivalent to apply \emph{model updating}~\cite{DMN06,DZ08,DZ06}.
\end{description}
}
%


\subsection{Setting of types and values of variables}

The system takes as input files representing MIMs as created by 
PathVisio\footnote{\url{https://github.com/PathVisio/pathvisio}},
  a free open-source biological pathway analysis software
that allows one to draw biological pathways. 
\hide{

The XML file produced by Pathvisio is taken as input by the system and
the graph is displayed to the user.

by our software and
displayed by a graphical interface.

Using Pathvisio can
also help us in retrieving already written graphs and using them for
our needs. Graphs are saved by Pathvisio in XML (GPML) format. This
means that graphs saved by Pathvisio can be read by any XML parser. 
}
The graph is displayed to the user, using colors and typefaces to
distinguish the types and initial values of atoms, which are
given a default value by the software tool based on ``commonsense''
rules.  
Figure~\ref{fig:lac2} shows how the software has set the variable types:
lacl, lacZ, CAMP and Lactose 
are in {\em bold} typeface, as they are set as {\em exogenous} variables, glucose,
galactosidase and repressor boxes are in {\em normal} typeface, as they are {\em
  endogenous}. 

\red{Variables initial values are shown by use of different
colors: by default, the initial values of all variables are {\em
  unset} and their names will be shown in black. Henceforth, 
atoms whose initial value is not set will be called
  {\em free}.}
\hide{
\begin{comment}
I have changed Figure~\ref{fig:lac2}, because now the software leaves
all values initially unset. And I have consequently modified the
sentence describing the figure.
\hide{
Figure~\ref{fig:lac2} shows how the software has set
types and values of variables: lacl, lacZ, CAMP and Lactose 
are in {\em bold} typeface, as they are set as {\em exogenous} variables, glucose,
galactosidase and repressor boxes are in {\em normal} typeface, as they are {\em
  endogenous}. 
}

Old Figure:
\begin{center}
\includegraphics[width=0.5\linewidth]{./graphics2/lac10008.png}
\end{center}

Old description:
The values of endogenous  atoms are initialized as absent
(red) and those of exogenous atoms are not set (black). Henceforth, 
atoms whose initial value is not set will be called
  {\em free}. Normally, 
all the endogenous variables values are initialized
as absent, while no initial value 
is set for the exogenous ones. In other terms, the default graph
corresponds to a MIG whose 
initial conditions $\cal B$ does not mention exogenous variables,
i.e. they do not occur, neither 
positive nor negated, in $\cal B$.
\end{comment}
}
\begin{figure}
\begin{center}
\includegraphics[width=0.6\linewidth]{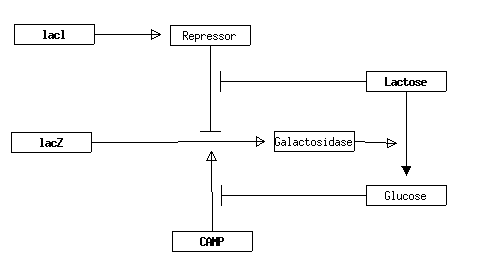}
\end{center}
\caption{The lac operon after the default initialization of variables types and
  values} 
\label{fig:lac2}
\end{figure}
The user is  allowed to change {both types and  initial values
  of atoms}. 
Figure~\ref{fig:lac3} shows the graph when  the user has modified the
values of some
variables: lacl, lacZ and CAMP are {\em green}, to indicate that they
are {\em present} at the start of the process (they will remain
present since they are exogenous atoms). Repressor is {\em green}, as
the repressor protein is supposed to be in the cell at the start of
the process. Lactose remains {\em black} since it is a free
atom, about which the user is going to query the system. Initially
absent variables (Galactosidase and Glucose) are shown in {\em red}.
\begin{figure}
\centering
\includegraphics[width=0.5\linewidth]{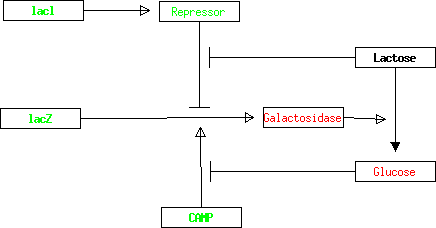}
\caption{Lac operon after the modifications}
\label{fig:lac3}
\end{figure}

\hide{
\begin{comment}
I wonder whether the description above is too detailed. Maybe it could
be better to replace it with the summary given in IJCAI submission.
I copy it below (maybe the references to figures etc are not
correct). When I have time, I'll try to merge the two descriptions.

 The graph is
displayed to the user, who can select the types and
values of the input variables. 
For instance, for the MIM in Figure~\ref{fig:mimlac}, 
the user is allowed to set types and values of variables as follows: lacl,
lacZ, CAMP and Lactose are 
set as {\em exogenous} variables while glucose, galactosidase and
repressor are set as {\em endogenous}. The system also ``guesses'' 
the initial variable 
values, showing them by use of different colors.  
Normally, 
all the endogenous variables values are initialized as absent, while
no initial value is set for the   exogenous ones. In other terms, the
default graph corresponds to a MIG whose initial conditions  $\cal B$
does not mention exogenous variables, i.e. they do not occur, neither
positive nor negated, in $\cal B$. The variables for which an
initial value is not set are called {\em free}. The user is also
allowed to change the initial variable values.

\end{comment}
}

Other parameters, such as the number of time steps, the number of
modifications to make for graph updating,  queries etc. are
 set via the command line. 


\subsection{Resolution engine}
The resolution engine is able to perform the following reasoning
tasks.

\begin{description}
\item[Graph validation.] 
\red{This task consists in checking whether the
  graph 
  ${\cal G}$ is consistent. The temporal encoding of ${\cal G}$ is
 grounded to the specified time and 
   the SAT solver Picosat is used in a straightforward way in
  order to check the consistence of the grounded theory. 
}
\hide{
This task consists in checking whether the
  graph 
  ${\cal G}$ satisfies some property expressed in temporal logic.
  In the implementation, the temporal encoding of ${\cal G}$ and the 
  negation of the property to be checked are grounded to the specified time. 
  Then the SAT solver Picosat is used in a straightforward way in
  order to check the consistence of the grounded theory. 
}

\hide{After having encoded the graph and performed the grounding to the
specified time \uline{of both the encoding and the negation of the property to be
checked},
 the SAT solver Picosat  is used in a straightforward way to check the consistency
of the so obtained grounded formulas. 

\begin{comment}
Please, check if I have understood well. The first presentation of this task
(that I have now hidden) says that: it consists in checking whether
  ${\cal G}$ satisfies some property expressed in temporal logic and
  it can be performed by doing \emph{model checking} on 
  the temporal translation. Then, when detailing the task, there was
  written that what is done is a 
  consistency check \uline{on the grounded formulae  describing the
    graph   and of the initial conditions and the final
    conditions}. Now, the graph includes the initial conditions, and I
  believe that the ``final conditions'' consist of the formula to be
  checked. But it has to be negated, hasn't it? 

Note also that in the IJCAI version this task is described as follows:
 the SAT solver is used in a straightforward way to check the consistency
of the grounded formulas describing the graph.

\end{comment}
}
\hide{
if only graph validation has to be done, then
  translation and grounding are directly performed (as described in
  section~\ref{sec:sat}) and a SAT solver~\cite{Picosat08} is
  used to check the consistency of the formulas describing the graph
  and of the initial conditions and the final conditions. The
  implementation is straightforward. 
}

\item [Graph querying.] This task consists in finding which initial
  values of the free atoms make ${\cal G}$ satisfy some temporal property
  $\varphi$. It is an abductive reasoning task~\cite{Inoue92}, that
 could be solved by use of
classical algorithms for computing prime implicants. But  we have
checked that, for instance, the Kean and
Tsiknis algorithm~\cite{Kean90} results to be very slow even when the
total number of atoms is small.
\hide{

as already stated, 
graph
  querying is equivalent to logical abduction.   If we call $D$ 
  the set of CNF formulas after grounding into propositional logic,
  and $Q$ the question, then
  we first check whether $D\cup \{\neg Q\}$ is
  consistent (if it is not consistent then $Q$ is already an implicate
  of $D$). Then
  we search for the minimal set $H$ such that $Q$ is an implicate of
  $T\cup H$. The classical algorithm consists in computing the set
  ${\cal  P}(D\cup \{\neg Q\})$, which is the set of the prime
  implicates (the strongest clausal consequences) of
  $D\cup \{\neg Q\}$, and then checking, for each
  $x \in {\cal P}(D\cup \{\neg Q\})$, whether $D\cup \{\neg x\}$ is
  consistent. Then each such $x$ is a solution.
This has been implemented by using the Kean and
Tsiknis algorithms~\cite{Kean90} for computing prime
implicates. However, this algorithm is slow, even when the number of
variables is not large, and its complexity depends on the total number of
 variables (both exogenous and endogenous). 
}
However, 
biologists are usually only interested by the values of the free
atoms. 
Since their number is often quite
small, it is usually faster to use Picosat to solve iteratively all
possible models. 
In other terms, all the possible combinations of initial values for
free atoms are generated (by the 
{\em formula enumerator} of figure~\ref{fig:implant}) and 
\tocheck
the SAT solver is run on
each of the so-obtained initial conditions.
The system,
tested on graphs with up to 22 nodes and 41 relations, showed to be
effective up to roughly 16 to 20 free 
atoms depending on the complexity of the map. 
\hide{
The generation of all the possible combinations of initial values for
free atoms 
is called the
{\em formula enumerator} in figure~\ref{fig:implant}.
All possible sets of values are tested and the SAT solver is run on
each of them.
}

\red{In performing this task,  exogenous and
  endogenous atoms can be  treated differently: the user can either
  ask which values of all the free variables imply the given property,
  or else to find out which values of the free exognenous atoms
  guarantee that {\em for all values of the free endogenous ones} the
  query holds at the given time.}

\item[Graph updating.] 
Given a graph 
${\cal G}$ for which a given property   $\varphi$ does not hold, this
task 
consists in turning ${\cal G}$ into
  a new graph ${\cal G'}$ satisfying $\varphi$.
\hide{
 task
consists in turning the graph ${\cal G}$ into
  a new one ${\cal G'}$ that explains a property
  $\varphi$ which is not explained
\begin{comment}{M: I changed satisfied by explained} 
\end{comment}
by the original ${\cal G}$. It is
  equivalent to apply \emph{model updating}~\cite{DMN06,DZ08,DZ06}
\begin{comment}{M: I am not so sure if this comment is appropiate
    here}.
\end{comment}
}
This  is the most complex task, since there is a very large number
of possible graphs solving the problem. Currently, the system computes
all graphs ${\cal G}'$ that can be obtained from ${\cal G}$ by adding, removing or
modifying a single  relation (this step is
called the {\em graph enumerator} in figure~\ref{fig:implant}). 
 Then for each ${\cal G}'$,
 graph querying on ${\cal G}'$ and $\varphi$ is performed, in order to filter
out those 
which do not satisfy $\varphi$.  .

\hide{
the system performs all
modifications that require modifying, adding or deleting one and
only one relation, and solves the graph querying problem for each of
them. This step is called the {\em graph enumerator} in
figure~\ref{fig:implant}. 
\begin{comment}
This description goes on saying: 

It
generates all possible graphs with only one modification starting from
the original graph. Then 
translation  and grounding are performed for
each of them. 
 \uline{It is also possible} (?) to perform graph querying for each
graph, using the formula enumerator. Then  the SAT solver is called
for each generated graph and 
each possible set of variables.

But it seems to me that all this things are already said before,
because the graph querying problem already performs encoding and
grounding and calls Picosat with each possible initial setting of free
atoms.
\end{comment}}

\end{description}

\section{Examples}\label{sec:examples}

The software tool has been tested on graphs with up to 20 atoms, 22 nodes and 41
links. In this section we show some
examples of 
the two most complex tasks: graph querying
and graph updating. 

\subsection{Graph querying}
\label{sec:abduction}
A more complex example will be considered here, {i.e.}, 
a meaningful  part of the map presented in
Gigure~\ref{fig:pommier},  
 the {\em atm-chk2} metabolic pathway, which leads to cellular
apoptosis when the DNA double strand breaks. 
DNA double strand break ({\em dsb}) is a major cause of cancers, and
medical and pharmaceutical research~\cite{Kohn:2005,Glorian:2011} have
shown that dsb can occur in a cell as the result of a pathology in a
metabolic pathway.
This kind of map is used to find the molecular determinants of tumoral
response to cancers. Molecular parameters included the metabolic
pathways for repairing DNA, the metabolic pathways for apoptosis, and
the metabolic pathways of cellular cycle
control~\cite{Pommier:2005,Kohn:2005,Glorian:2011,Lee:2007,peimmset2011}.
When DNA is damaged, cellular cycle control points are activated and
can quickly kill the cell by apoptosis, or stop the cellular cycle to
enable DNA repair before reproduction of cellular division. Two of
these control points are the metabolic pathways {\em atm-chk2} and
{\em atr-chk2}~\cite{Pommier:2005}.

The graph of Figure~\ref{fig:atmchk2} (built from the map in
Figure~\ref{fig:pommier}) represents the metabolic
pathway {\em atm-chk2} which can lead to apoptosis in three different
ways.
This map involves 20 variables, six of which ({\em atm, dsb, chk2,
mdm2, pml} and {\em p53}) are exogenous and the 
rest endogenous.
Some of these variables are proteins, others, such as {\em dsb} or {\em apoptose}, representing  cell
death, are conditions or states.

\begin{figure}
\centering
\includegraphics[width=\linewidth]{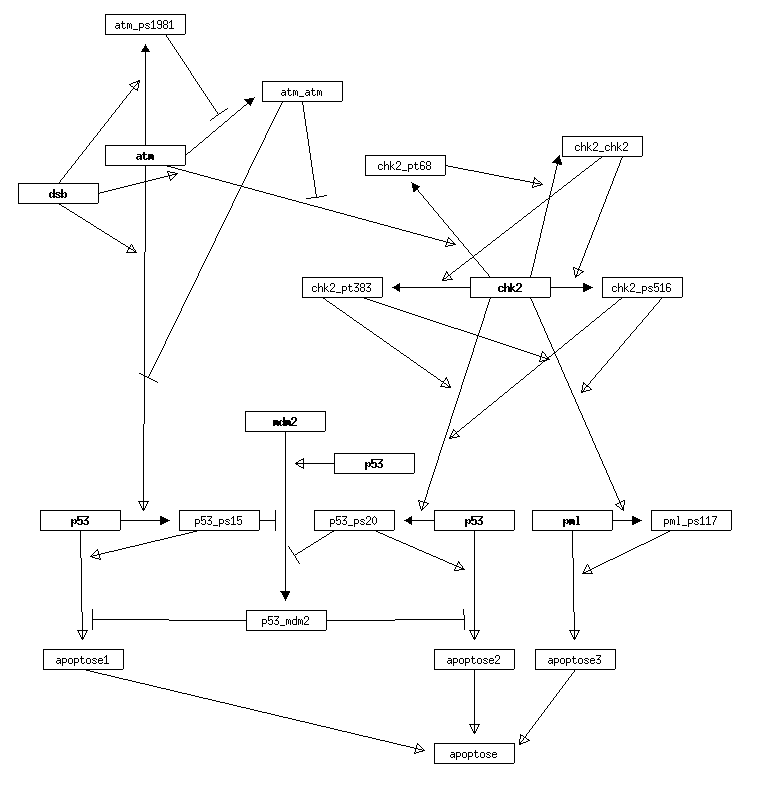}
\caption{The Molecular Interaction Map {\em atm-chk2}}
\label{fig:atmchk2}
\end{figure}

The time required for solving graph querying problems depends on the
number of free variables and  time steps. 
\hide{
considering 
$t$ time steps, there are $20+42t$ variables and $14+145t$ clauses.
If
we have $n$ free variables, the SAT solver is called $2^n$
times, because only the values of free variables are interesting for
the user.  }
The P3M solver has been called on this graph to find out  what would cause 
the cell
apoptosis. It has been tested with different grounding
values $g$, ranging from 1 to 50, and queries to find out the
initial conditions that make the atom $apoptose(g)$ derivable, i.e., 
the conditions causing  cell apoptosis at time $g$.
The system has been tested  with a number of
free variables ranging from 6 (only exogenous variables are free)
to 20 (all  variables are set to free, thus asking the
system to find also their initial values). 

\hide{
\begin{figure}
\begin{center}
  \subfloat[Time as a function of the number of grounding steps and free variables]{\includegraphics[width=\linewidth]{graphics2/timings.png}}\\
  \subfloat[Time for 100 grounding steps]{\includegraphics[width=0.4\linewidth]{graphics2/timings2.png}}
  \subfloat[Time for 20 free variables]{\includegraphics[width=0.4\linewidth]{graphics2/timings3.png}}
\end{center}
\caption{Timings}
\label{fig:timings}
\end{figure}
}

\begin{figure}
\begin{center}
  {\includegraphics[width=\linewidth]{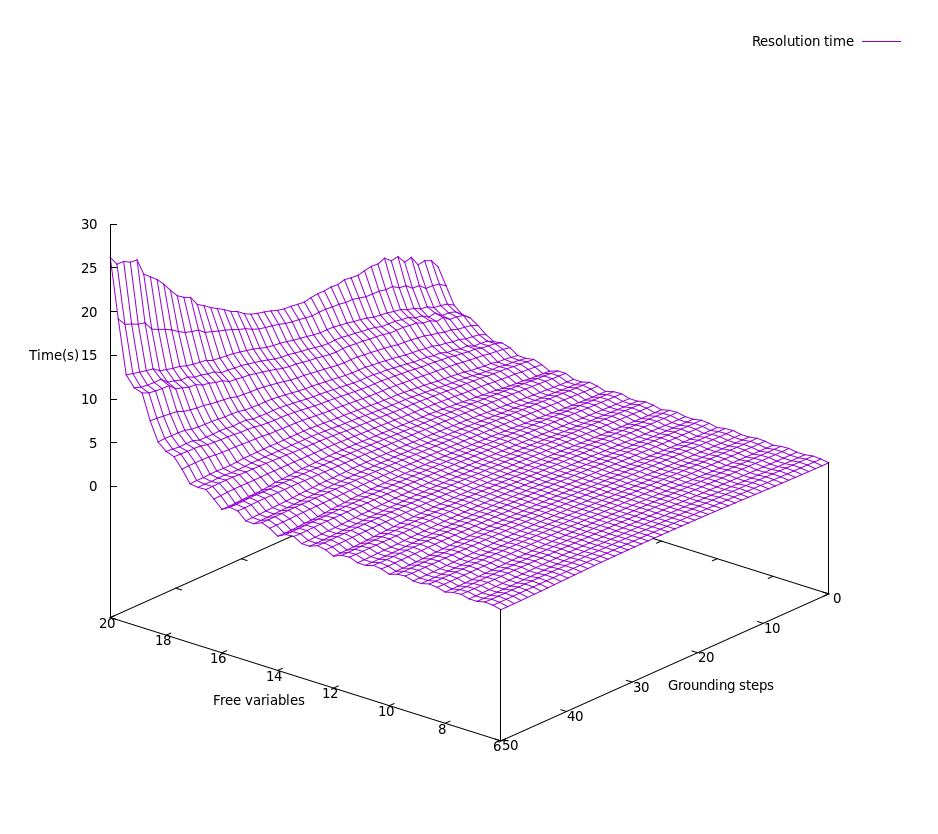}}
\end{center}
\caption{Time as a function of the number of grounding steps and free variables}
\label{fig:timings}
\end{figure}

The 3D diagram in Figure \ref{fig:timings} plots the grounding values
and the number of free variables 
against the time taken by the system to solve the problem, by
calling Picosat (the time taken to encode the graph into propositional logic is negligible).
 From the diagram, it is clear
that the number of free variables is the bottleneck, as it was actually
expected since the time required to solve the problem is exponential in
the number of free variables.  Moreover, 50 time steps are overkill, most systems reaching a
stable state in less than 10 time steps.

\hide{
The size of the set of
  propositional formulae encoding the problem is proportional to the
  grounding value. The diagram in Figure \ref{fig:cnf} shows that 
the size of the encoding grows with the total time with the same rate as
the grounding value does. When considering the weight of  solving
time {w.r.t.}  total time, it must also be taken into account that,
in this example, a solution can be found already with grounding 3,
therefore the difficulty of the 
reasoning task does not increase significantly with the
grounding values.
\begin{comment}
Maybe this should be put in relation with Theorem 
\ref{encoding-complexity}
\end{comment}

\begin{figure}[htb]
\begin{center}
\includegraphics[width=\columnwidth]
{graphics2/ex3-cnf.pdf}
\end{center}
\caption{Size of the propositional encoding of the problem}
\label{fig:cnf}
\end{figure}
}

The questions asked to the system can be refined, in order to find
out, for instance,  how much time is required to reach apoptosis on
each of the three possible ways, and which are the initial conditions
which lead to each of them.  The questions to ask are {\em apoptose1(i)},
{\em apoptose2(i)} and {\em  apoptose3(i)}, for different values of  $i$, where a
\tocheck
query of the form $p(i)$ means that one looks for an explanation of
$p$ being true at time step $i$.
The answers given by the system  show that:
\begin{itemize}
\item
  {\em apoptose1} can be obtained is the in the fastest way: 
\tocheck
{\em apoptose1(2)} ({\em apoptose1}
  holding  at the second time
  step) is true if {\em atm}, {\em dsb} and {\em p53} are present, and
  {\em mdm2} is
  absent (the values of {\em pml} and {\em chk2} do not matter).  For $i\geq 3$,
  the answer to {\em apoptose1(i)} is the same, but {\em mdm2} does not matter
  any longer ({\em p53\_mdm2} is dissociated at step 2).
\item obtaining
  {\em apoptose2} requires 5 time steps; 
{\em atm}, {\em chk2}, {\em dsb}, {\em p53} have
  to be present, and {\em mdm2} and {\em pml} do not matter. 
\item
  {\em apoptose3} requires the same number of steps as {\em apoptose2} but
  the initial 
conditions are different: {\em atm}, {\em chk2}, {\em dsb}, and {\em pml} have to be
  present, while {\em mdm2} and {\em p53} do not matter.
\end{itemize}

\subsection{Graph updating}
\label{sec:mise}
Figure~\ref{fig:lac1b1} shows the map of the lac operon where 
 the inhibition of lactose on the negative regulation of the
repressor to the production of galactosidase has been suppressed. 
So here, glucose is not
produced anymore when lactose is present. 
\begin{figure}
\centering
\includegraphics[scale=0.5]{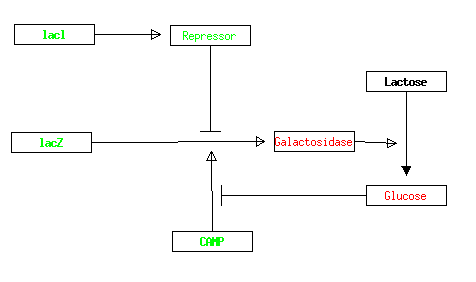}
\caption{Lac operon without inhibition by lactose}
\label{fig:lac1b1}
\end{figure}
The user can ask the system what modifications could be done in
order to produce glucose when lactose is present. 
The ``correct'' solution is found immediately (Figure~\ref{fig:lac1b2}),
along with others. Some of these other generated solutions have no interest,
such as the direct production of glucose by genes lacZ or lacl. 
But the system also proposes reasonable solutions, such as that shown
in Figure~\ref{fig:lac1b3}, where glucose is used to provide the inhibiting
action for the repressor protein. When glucose is present, the
production of galactosidase is stopped, while it is done when glucose
is absent. However nature has chosen the more economical solution,
because here galactosidase would be produced as soon as glucose is
absent, which is useless if there is no lactose. 

\begin{figure}[htb]
    \centering
    \includegraphics[scale=2]{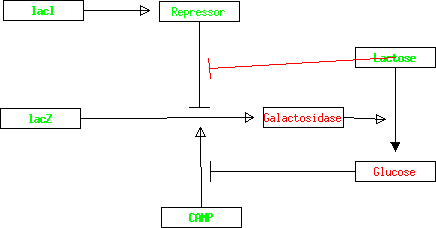}
    \caption{Correct solution}
    \label{fig:lac1b2}
\end{figure}

\begin{figure}[htb]
    \centering
    \includegraphics[scale=2]{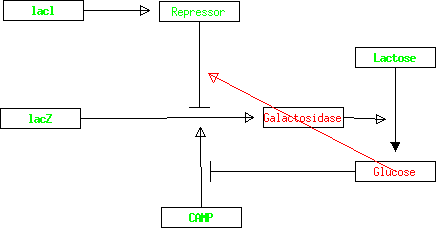}
    \caption{Another interesting solution}
    \label{fig:lac1b3}
\end{figure}

\section{Conclusion}
\label{sec:conclusion}
This paper presents a method to translate MIMs,
representing biological systems, into Linear Temporal Logic, and a software
tool able to
solve complex questions on these graphs. The system, though still a 
prototype, is able to solve quite realistic examples of a
large size. 

The proposed approach can be improved in different directions. On the
theoretical side, it is worth remarking that the speed of reactions is
not taken into account. This limitation could be overcome by using the
dual of speed (duration) and by using a logic that represents the
duration of reactions. Moreover, the system relies on the ``all or
nothing'' hypothesis: we do not represent quantities other than
``absent'' or ``present''. 
As a consequence, all productions that
\tocheck
are enabled at a given time are fired simultaneously, since they do
not compete on the use of resources.
Even if we have been able to efficiently model 
complex graphs with this constraint, an important step forward to be
planned is modelling a more realistic evolution of networks by taking
quantities into account. 

On the practical point
of view, the possibility should be explored to avoid grounding and
\tocheck
replacing the 
{\em formula enumerator} procedure of P3M by
implementing a direct abduction algorithm for (a suitable fragment of)
LTL, as proposed in \cite{CP17}, or else 
by directly using  temporal model checkers \cite{clarke-2003}, or
tools like RECAR 
(Recursive Explore and Check Abstraction Refinement ) \cite{lagniez-17} which
allows one to solve modal satisfiability problems . 

 Moreover, 
the software tool can be improved in several respects like,
for instance, 
improving the graphical
interface by enriching the number of parameters the user can choose
and making it more user friendly.



\bibliographystyle{plain}
\bibliography{refs}

\end{document}